\documentclass[10pt, journal]{IEEEtran}

\usepackage{graphicx}
\usepackage{booktabs} 
\usepackage{longtable}
\usepackage{wrapfig}
\usepackage{amsmath}
\usepackage{amssymb}
\usepackage{capt-of}
\usepackage{comment}
\usepackage{dsfont}
\usepackage{url}

\usepackage{enumitem}
\usepackage{bm}

\usepackage{amsmath}
\usepackage{amssymb}

\usepackage{mathtools}
\usepackage{amsthm}

\usepackage{xcolor}
\usepackage{bm}
\usepackage{subcaption}
\usepackage{algorithm}
\usepackage{algorithmic}

\newtheorem{theorem}{Theorem}[section]

\newtheorem{lemma}{Lemma}[section]
\newtheorem{corollary}{Corollary}[section]
\theoremstyle{definition}
\newtheorem{definition}{Definition}

\theoremstyle{remark}
\newtheorem{remark}{Remark}[section]

\newcommand{\norm}[1]{\left\lVert #1 \right\rVert}

\usepackage{authblk}

\begin{document}

\title{On the Price of Differential Privacy for Spectral Clustering over Stochastic Block Models}

\author{Antti~Koskela$^{*}$, Mohamed~Seif$^{*}$, Andrea~J.~Goldsmith
\thanks{$^{*}$A. Koskela and M. Seif contributed equally to this work.}%
\thanks{A. Koskela is with Nokia Bell Labs.}%
\thanks{M. Seif and A. J. Goldsmith are with Princeton University.}%
\thanks{ The work was supported by the AFOSR award \#002484665,  Huawei Intelligent Spectrum grant, and NSF award CNS-2147631. The authors thank Iraj Saniee for facilitating this collaboration.
    }
}

\markboth{}
{Shell \MakeLowercase{\textit{et al.}}: Bare Demo of IEEEtran.cls for Computer Society Journals}

\maketitle

\begin{abstract}
We investigate privacy-preserving spectral clustering for community detection within stochastic block models (SBMs). Specifically, we focus on edge differential privacy (DP) and propose private algorithms for community recovery. Our work explores the fundamental trade-offs between the privacy budget and the accurate recovery of community labels. Furthermore, we establish information-theoretic conditions that guarantee the accuracy of our methods, providing theoretical assurances for successful community recovery under edge DP.
\end{abstract}

\hspace{0.1in}

\begin{IEEEkeywords}
Differential Privacy, Graphs, Stochastic Block Model, Perturbation, Community Detection, Spectral Clustering.
\end{IEEEkeywords}

\section{Introduction}
\label{sec:introduction}

Community detection within networks is a pivotal challenge in graph mining and unsupervised learning \cite{fortunato2010community}. The primary objective is to identify divisions (or communities) in a graph where connections are densely concentrated inside communities and sparsely distributed between them. The Stochastic Block Model (SBM) is widely utilized to represent the structural patterns of networks \cite{abbe2018community}. In the SBM framework, nodes are assigned to specific communities, and the probability of connections between any two nodes is based on their community memberships. Specifically, nodes within the same community are more likely to be connected than those in different communities. This variation in connection probabilities is fundamental to the challenge of detecting communities. Research aimed at studying and enhancing community detection methods using the SBM approach has been highly active, with numerous advancements and discoveries detailed in comprehensive reviews such as the one by Abbe et al. \cite{abbe2017community}.

Network data, such as the connections found in social networks, often contain sensitive information. Therefore, protecting individual privacy during data analysis is essential. Differential Privacy (DP) \cite{dwork2014algorithmic} has become the standard method for providing strong privacy guarantees. DP ensures that the inclusion or exclusion of any single user's data in a dataset has only a minimal effect on the results of statistical queries.

In the realm of network or graph data, both edge and node privacy models have been investigated. As discussed in \cite{karwa2011private}, two primary privacy concepts have been introduced for analyzing graph data: (1) Edge DP, which aims to safeguard individual relationships (edges) within a graph by utilizing randomized algorithms to minimize the impact of any specific edge’s presence or absence during analysis, and (2) Node DP, which focuses on protecting the privacy of nodes and their associated connections (edges). Edge DP is better suited for private community detection, as it focuses on protecting individual relationships, which are central to defining and identifying community labels. Additionally, DP algorithms have been adapted to address specific network analysis tasks, such as counting stars, triangles, cuts, dense subgraphs, and communities, as well as generating synthetic graphs \cite{nguyen2016detecting, qin2017generating, imola2021locally, blocki:itcs13}. More recently, in our previous work, we explored community detection  in SBMs under the edge privacy model for various settingsand established sufficient conditions for recoverability thresholds using ML-based estimators and their semidefinite relaxations  \cite{mohamed2022differentially,  seif2024differentially,  seif2024private}.

In the existing literature, efficient algorithms for community detection in SBMs have been developed using spectral methods, such as those detailed in \cite{dhara2022power, dhara2022spectral}, as well as through semidefinite programming (SDP) approaches \cite{hajek2016achieving}. While these spectral methods have demonstrated significant effectiveness in terms of the computational complexity in identifying community structures, there remains a limited understanding of their performance under privacy constraints. Specifically, there is a notable gap in knowledge regarding private spectral methods for various community recovery requirements, including partial and exact recovery. This highlights an important area for future research aimed at ensuring that community detection techniques can both preserve privacy and maintain high accuracy across different recovery requirements (e.g., exact or partial recovery) within the SBM framework.

\textbf{Related Work.}  The closest related work to our study is~\cite{hehir2021consistency}, in which the authors analyze the consistency of privacy–preserving spectral clustering under the Stochastic Block Model (SBM). While insightful, their results stop short of deriving \emph{explicit separation conditions} that tie together the SBM parameters (e.g., block edge probabilities, number of communities) and the privacy budget~$\varepsilon$. Pinpointing these conditions is essential for understanding \emph{when} private spectral methods can provably recover communities and \emph{how} the privacy constraint degrades the signal–to–noise ratio required for success.

Our earlier efforts~\cite{mohamed2022differentially,seif2024differentially} addressed this question from an optimization standpoint by casting community detection as a semidefinite program (SDP). Although the SDP approach delivers strong statistical guarantees, its polynomial‑time complexity renders it impractical for the massive graphs that arise in modern social‑network or e‑commerce platforms—often containing tens of millions of nodes and billions of edges.

\textbf{Key Outstanding Challenge.} A rigorous privacy--utility trade‑off analysis is still needed—one that links scalable private algorithms to explicit separation thresholds expressed in terms of SBM parameters and the privacy budget~$\varepsilon$. Closing this theoretical gap would pave the way for deployable, high‑accuracy, privacy‑preserving community detection that (i) \textbf{scales} to multi‑million‑node graphs, (ii) \textbf{respects} user‑level differential‑privacy guarantees, and (iii) \textbf{achieves} the full spectrum of recovery objectives—exact, partial, or weak consistency—studied in the SBM literature.

\textbf{Contributions.} We make the following key contributions to privacy-preserving spectral clustering under edge DP for community detection on the symmetric binary SBM\footnote{Generalizing the privacy mechanisms to multiple communities is a sufficiently interesting direction and is left for future work.}:

\begin{enumerate}
    \item \textbf{Graph Perturbation-Based Mechanism}: We apply the randomized response technique to perturb the adjacency matrix of the graph. Subsequently, a spectral clustering algorithm is executed on the perturbed graph to recover community structures. This approach inherently satisfies $\epsilon$-DP for any $\epsilon > 0$ due to the post-processing property of differential privacy.
    
    \item \textbf{Subsampling Stability-Based Mechanism}: Inspired by the work of \cite{dwork2014algorithmic}, we introduce the subsampling stability-based estimator, which involves generating multiple correlated subgraphs by randomly sampling edges with probability $q_s$. A non-private clustering algorithm is then applied to each subgraph, and the resulting community labels are aggregated into a histogram. The stability of this histogram, influenced by parameters such as $(p, q)$, $(\epsilon, \delta)$, $q_s$, and the number of subgraphs $m$, ensures accurate recovery of the original community labels.

    \item \textbf{Noisy Power Iteration Method}: We execute the power method while injecting carefully calibrated Gaussian noise at every matrix–vector multiplication to obfuscate each edge’s contribution. Subsequently, the normalized noisy eigenvectors are used to form the clustering embedding. This approach inherently satisfies $(\epsilon,\delta)$-edge DP for any $\epsilon > 0$ (with suitably small $\delta$) via the Gaussian mechanism and iterative composition.
    
    \item \textbf{Tradeoff Analysis and Theoretical Guarantees}: We investigate the fundamental tradeoffs between the privacy budget and the accuracy of community recovery. Additionally, we provide theoretical guarantees by establishing information-theoretic conditions that ensure successful community detection under edge DP. Furthermore, we provide a lower bound on the overlap rate between the estimated labels and the ground truth labels of the communities.
\end{enumerate}

\textbf{Notation.} Boldface uppercase letters denote matrices (e.g., $\textbf{A}$), while boldface lowercase letters are used for vectors (e.g., $\textbf{a}$). We use $\operatorname{Bern}(p)$ to denote a Bernoulli random variable with success probability $p$.
For asymptotic analysis, we say the function $f(n) = o(g(n))$ when $\lim_{n \rightarrow \infty} f(n)/g(n) = 0$.  Also, $f(n) = {O}(g(n))$ means there exist some constant $C > 0$ such that $|f(n)/g(n)| \leq C$, $\forall n$, and $f(n) = \Omega(g(n))$ means there exists some constant $c > 0$ such that $|f(n)/g(n)| \geq c$, $\forall n$.

\section{Problem Statement \& Preliminaries}
\label{sec:preliminaries_and_problem_statement}

 We consider an undirected graph $G = ({V}, E)$ consisting of $n$ vertices, where the vertices are divided into two equally sized communities $C_{1}$ and $C_{2}$,  ${V} = C_{1} \cup C_{2}$ and {{$E$ is the edge set}}. The community label for vertex $i$ is denoted by $\sigma^{*}_{i} \in \{-1, +1\}, \forall i \in [n]$. Further, we assume that the graph $G$ is generated through an SBM, where the edges within the same community are generated independently with probability $p$, and the edges across the communities $C_{1}$ and $C_{2}$ are generated independently with probability $q$. The connections between vertices are represented by an adjacency matrix $\mathbf{A} \in \{0, 1\}^{n \times n}$, where the elements in $\mathbf{A}$ are drawn as follows:
 \begin{align*}
     A_{i,j} \sim \begin{cases} \operatorname{Bern}(p), & i < j, ~~\sigma^{*}_{i} = \sigma^{*}_{j},  \\
     \operatorname{Bern}(q), & i < j, ~~\sigma^{*}_{i} \neq \sigma^{*}_{j} 
     \end{cases}
 \end{align*}
with $A_{i,i}=0$ and $A_{i,j}=A_{j,i}$ for $i > j$.

\begin{definition}  [Laplacian Matrix] Let $G = (V, E)$ be undirected graph. The Laplacian $\mathbf{L} \in \mathds{R}^{n \times n}$ of $G$ is the matrix defined as 
\begin{align*}
    \mathbf{L} & = \mathbf{D} - \mathbf{A},
\end{align*}
where \( \mathbf{D} \in \mathbb{R}^{n \times n} \) is the degree matrix, which is a diagonal matrix where each diagonal entry \( d_{ii} \) is defined as $d_{ii} = \sum_{j=1}^{n} A_{ij}$, and all off-diagonal entries of \( \mathbf{D} \) are zero.
\label{definition:laplacian_matrix}
\end{definition}

\textbf{Community Detection via Spectral Method.} Spectral clustering partitions the vertices of a graph \( G \) into communities by leveraging its spectral properties. To accomplish this, we perform an eigen decomposition of the Laplacian matrix \( \mathbf{L} \), obtaining its eigenvalues  $0 = \lambda_{1} \leq \lambda_{2} \leq \cdots \leq \lambda_{n}$ and their corresponding eigenvectors $\mathbf{u}_{1}, \mathbf{u}_{2}, \ldots, \mathbf{u}_{n}$. In the case of dividing the graph into two communities, spectral clustering specifically utilizes the eigenvector associated with the second smallest eigenvalue \( \lambda_{2} \) of \( \mathbf{L} \). This eigenvector effectively captures the essential structure needed to separate the graph's vertices into distinct communities based on the graph's connectivity \cite{von2007tutorial}.

\begin{definition}[$(\beta, \eta)$-Accurate Recovery]
A community recovery alghorithm $\hat{\bm{\sigma}}(G) = \{\hat{\sigma}_{1}, \hat{\sigma}_{2}, \cdots, \hat{\sigma}_{n}\}$ achieves $(\beta,\eta)$-accurate recovery (up to a global flip) if
\begin{equation}\label{eqn:beta_eta_accurate_definition}
\operatorname{Pr} \Bigl(\operatorname{err}\operatorname{rate}\bigl(\hat{\bm{\sigma}}(G), \bm{\sigma}^{*}\bigr) \leq \beta\Bigr) \geq 1 - \eta,
\end{equation}
where the probability is taken over both the randomness of the graph $G$ (drawn according to an SBM) and the randomness of the algorithm. Here, the error rate (up to a global flip) is defined via the Hamming distance as
\begin{equation*}
\operatorname{err}\operatorname{rate}\bigl(\hat{\bm{\sigma}}(G), \bm{\sigma}^{*}\bigr)
\;=\;
\frac{1}{n} \cdot 
\min_{s \in \{+1, -1\}} \operatorname{Ham}\bigl(\hat{\bm{\sigma}}(G), s \bm{\sigma}^{*}\bigr).
\end{equation*}
\end{definition}

 \begin{definition} [$(\epsilon, \delta)$-edge DP] \label{def:edgeDP}  A (randomized) community estimator $\hat{\bm{\sigma}}$ as a function of $G$ satisfies $(\epsilon, \delta)$-edge DP for some $\epsilon \in \mathds{R}^{+}$ and $\delta \in [0,1]$, if for all pairs of adjacency matrices $G$ and $G'$ that differ in {\it one} edge, and any measurable subset $\mathcal{S} \subseteq \operatorname{Range}(\hat{\bm{\sigma}})$, we have 
\begin{align*}
    {\operatorname{Pr}(\hat{\bm{\sigma}}(G) \in \mathcal{S}}) \leq e^{\epsilon} { \operatorname{Pr}(\hat{\bm{\sigma}}(G') \in \mathcal{S}}) + \delta,
\end{align*}
where the probabilities are computed only over the randomness in the estimation process. The setting when $\delta = 0$ is referred as pure $\epsilon$-edge DP
\end{definition} 

\begin{algorithm}[t]
  \caption{Spectral Clustering Algorithm}
  \label{algo:spectral_clustering}
  \begin{algorithmic}[1]
    \STATE \textbf{Input:} \( G(\mathcal{V}, E) \)
    \STATE \textbf{Output:} Labeling vector \( \hat{\bm{\sigma}}(\mathbf{A}) \)
    \STATE Compute Laplacian: \( \mathbf{L} = \mathbf{D} - \mathbf{A} \)
    \STATE Eigen decomposition of \( \mathbf{L} \): obtain \( \lambda_i, \mathbf{u}_i \)
    \STATE Select Fiedler vector \( \mathbf{u}_2 \) for \( \lambda_2 \)
    \STATE Assign communities:
    \[
    \hat{\sigma}(v) = 
    \begin{cases} 
      1 & \text{if } u_{2,v} \leq 0 \\
     -1 & \text{otherwise} 
    \end{cases}
    \]
    \STATE \textbf{Optional:} Flip labels to minimize clustering error
    \STATE \textbf{Return:} \( \hat{\bm{\sigma}}(\mathbf{A}) \)
  \end{algorithmic}
\end{algorithm}

\section{Main Results \& Discussions} 
\label{sec:main_results}

We first establish a lower bound for all DP community recovery algorithms applied to graphs generated from binary SBMs, utilizing packing arguments under DP \cite{vadhan2017complexity}.
We then consider three different DP community recovery algorithms: graph perturbation-based mechanism, subsampling stability-based mechanism and a noisy power method applied on the adjacency matrix.

\subsection{General Lower Bound for $\epsilon$-edge DP Community Recovery Algorithms}

We establish a rigorous lower bound for all differentially private community recovery algorithms operating on graphs generated from SBMs. Our methodology closely follows the frameworks outlined in \cite{vadhan2017complexity, chen2023private}, focusing on the notion of edge DP. Precisely speaking, we define the classification error rate as
\begin{equation*}
\begin{aligned}
 & {\operatorname{err}} \operatorname{rate} (\hat{\bm{\sigma}}(\mathbf{A}), {\bm{\sigma}^{*}}) \\ & =  \frac{1}{n} \cdot \min \{\operatorname{Ham}(\hat{\bm{\sigma}}(\mathbf{A}), {\bm{\sigma}^{*}}), \operatorname{Ham}(- \hat{\bm{\sigma}}(\mathbf{A}), {\bm{\sigma}^{*}}) \}.
\end{aligned}
\end{equation*}

Let us consider a series of pairwise disjoint sets \(\mathcal{S}_{i}, i\in[m]\). Each set \(\mathcal{S}_{i}\) contains vectors \(\bm{u} \in \{\pm 1\}^n\), where \(n\) is the vector dimension. A vector \(\bm{u}\) is included in \(\mathcal{S}_{i}\) if 
\(\operatorname{err}\operatorname{rate}(\bm{u}, \bm{\sigma}^{i})\) with a fixed vector \(\bm{\sigma}^{i}\) does not exceed the threshold \(\beta\). This is formally expressed as: 
\begin{equation} \label{eq:S_i_def}
\begin{aligned}
\mathcal{S}_{i} = \{ \bm{u} \in \{\pm 1\}^n : \operatorname{err}\operatorname{rate}(\bm{u}, \bm{\sigma}^{i} ) \leq \beta \},
\end{aligned}
\end{equation}
$i = 1, 2, \ldots, m$, where $\mathcal{S}_{i}$'s are pairwise disjoint sets.

We next derive the necessary conditions for 
\begin{equation} \label{eqn:utility_condition_recovery}
\begin{aligned}
    {\operatorname{Pr}}(\hat{\bm{\sigma}}(\mathbf{A}) \in \mathcal{S}_{i}) \geq 1 - \eta 
\end{aligned}
\end{equation}
as a function of the SBM parameters and the privacy budget. Note that the randomness here is taken over the randomness of graph $G$ that is generated from $\operatorname{SBM} (\bm\sigma^{i},n, p, q)$. 

\begin{lemma} \label{lem:intermediate1}
Let $\bm\sigma^{i}$ be a fixed vector and the set $\mathcal{S}_{i}$ be defined as in Eq.~\eqref{eq:S_i_def} for some $\beta > 0$.
Suppose the condition of Eq.~\eqref{eqn:utility_condition_recovery} holds. Then,
\begin{equation*} 
\begin{aligned}
  (1-\eta)^{2}  & \leq \left( \mathds{E}_{\mathbf{A}, \mathbf{A}' \sim \Pi(\mathbf{A}, \mathbf{A}')} \left[ e^{2 \epsilon \operatorname{Ham}(\mathbf{A}, \mathbf{A}')} \right] \right)  \\
  &\times  \operatorname{Pr}(\hat{\bm{\sigma}}(\mathbf{A'}) \in \mathcal{S}_{i}) , 
\end{aligned}
\end{equation*}
\end{lemma}

\begin{proof}

Without loss of generality, let us consider a graph $\mathbf{A} \sim \operatorname{SBM} (\bm\sigma^{1},n, p, q)$ that is generated from the ground truth labeling vector $\bm{\sigma}^{*} = \bm{\sigma}^{1}$. Further, for this case, we want to derive the necessary conditions for any $\epsilon$-edge DP recovery algorithms that
\begin{equation*}
        {\operatorname{Pr}}(\hat{\bm{\sigma}}(\mathbf{A}) \in \mathcal{S}_{1}) \geq 1 - \eta 
\end{equation*}
which implies that
\begin{equation} \label{eqn:equation_necessary_conditions}
\begin{aligned}
 \sum_{i = 2}^{m}   {\operatorname{Pr}}(\hat{\bm{\sigma}}(\mathbf{A}) \in \mathcal{S}_{i}) \leq \eta, 
\end{aligned}
\end{equation}
where $\mathbf{A} \sim \operatorname{SBM} (\bm\sigma^{1},n, p, q)$.

We next individually lower bound each term in Eqn. \eqref{eqn:equation_necessary_conditions}. To do so, we first invoke the group privacy property of DP \cite{dwork2014algorithmic} and show that for any two adjacency matrices $\mathbf{A}$ and $\mathbf{A}'$, we have
\begin{align*}
    \operatorname{Pr}(\hat{\bm{\sigma}}(\mathbf{A}) \in \mathcal{S}) \leq e^{\epsilon \operatorname{Ham}(\mathbf{A}, \mathbf{A}')}     \operatorname{Pr}(\hat{\bm{\sigma}}(\mathbf{A'}) \in \mathcal{S}), 
\end{align*}
for any measurable set $\mathcal{S} \subseteq \{\pm 1\}^{n}$. For each $i = 1,2, \cdots, m$, taking the expectation with respect to the coupling distribution $\Pi(\mathbf{A}, \mathbf{A}')$ between $\mathbf{A}$ and $\mathbf{A}'$ and setting $\mathcal{S} = \mathcal{S}_{i}$, yields the following:
\begin{equation*} 
\begin{aligned}
      &\mathds{E}_{\mathbf{A}, \mathbf{A}' \sim \Pi(\mathbf{A}, \mathbf{A}')} \left[ \operatorname{Pr}(\hat{\bm{\sigma}}(\mathbf{A}) \in \mathcal{S}_{i}) \right] \\ & \leq \mathds{E}_{\mathbf{A}, \mathbf{A}' \sim \Pi(\mathbf{A}, \mathbf{A}')} \left[ e^{\epsilon \operatorname{Ham}(\mathbf{A}, \mathbf{A}')}     \operatorname{Pr}(\hat{\bm{\sigma}}(\mathbf{A'}) \in \mathcal{S}_{i})\right]  
\end{aligned}
\end{equation*}
which implies that
\begin{equation*} 
\begin{aligned}
& \operatorname{Pr}(\hat{\bm{\sigma}}(\mathbf{A}) \in \mathcal{S}_{i}) \\
 \leq & \mathds{E}_{\mathbf{A}, \mathbf{A}' \sim \Pi(\mathbf{A}, \mathbf{A}')} \left[ e^{\epsilon \operatorname{Ham}(\mathbf{A}, \mathbf{A}')}     \operatorname{Pr}(\hat{\bm{\sigma}}(\mathbf{A'}) \in \mathcal{S}_{i})\right]  \\ 
  \overset{(a)} \leq &  \left( \mathds{E}_{\mathbf{A}, \mathbf{A}' \sim \Pi(\mathbf{A}, \mathbf{A}')} \left[ e^{2 \epsilon \operatorname{Ham}(\mathbf{A}, \mathbf{A}')} \right] \right)^{1/2} \\
 & \times  \left(\mathds{E}_{\mathbf{A}, \mathbf{A}' \sim \Pi(\mathbf{A}, \mathbf{A}')} \left[ \operatorname{Pr}^{2}(\hat{\bm{\sigma}}(\mathbf{A'}) \in \mathcal{S}_{i}) \right] \right)^{1/2} \\ 
  \leq & \left( \mathds{E}_{\mathbf{A}, \mathbf{A}' \sim \Pi(\mathbf{A}, \mathbf{A}')} \left[ e^{2 \epsilon \operatorname{Ham}(\mathbf{A}, \mathbf{A}')} \right] \right)^{1/2} \\
  &\times \left( \operatorname{Pr}(\hat{\bm{\sigma}}(\mathbf{A'}) \in \mathcal{S}_{i}) \right)^{1/2}, \\ 
\end{aligned}
\end{equation*}
from which we further get that
\begin{equation*} 
\begin{aligned}
  1-\eta  & \overset{(b)} \leq \left( \mathds{E}_{\mathbf{A}, \mathbf{A}' \sim \Pi(\mathbf{A}, \mathbf{A}')} \left[ e^{2 \epsilon \operatorname{Ham}(\mathbf{A}, \mathbf{A}')} \right] \right)^{1/2} \\
  & \times \left( \operatorname{Pr}(\hat{\bm{\sigma}}(\mathbf{A'}) \in \mathcal{S}_{i}) \right)^{1/2}, \\ 
\end{aligned}
\end{equation*}
and
\begin{equation} \label{eqn:privacy_utility_condition}
\begin{aligned}
  (1-\eta)^{2}  & \leq \left( \mathds{E}_{\mathbf{A}, \mathbf{A}' \sim \Pi(\mathbf{A}, \mathbf{A}')} \left[ e^{2 \epsilon \operatorname{Ham}(\mathbf{A}, \mathbf{A}')} \right] \right)  \\
  &\times  \operatorname{Pr}(\hat{\bm{\sigma}}(\mathbf{A'}) \in \mathcal{S}_{i}) , 
\end{aligned}
\end{equation}

where step (a) follows from applying Cauchy-Schwartz inequality. In step (b), we invoked the condition in Eqn. \eqref{eqn:utility_condition_recovery}.

\end{proof}

We next focus on computing the term $\mathds{E}_{\mathbf{A}, \mathbf{A}' \sim \Pi(\mathbf{A}, \mathbf{A}')} \left[ e^{2 \epsilon \operatorname{Ham}(\mathbf{A}, \mathbf{A}')} \right]$ in the expression of Lemma~\ref{lem:intermediate1}. This together with 
Lemma~\ref{lem:intermediate1} leads to the general lower bound for the number of vertices in the graph.



It is worthwhile mentioning that the Hamming distance between two labeling vectors $\bm{\sigma}$ and $\bm{\sigma}'$ (each of size $n$) directly determines how many rows in the adjacency matrices $\mathbf{A}$ and $\mathbf{A}'$ are generated from the same versus different distributions. More precisely, we have two cases: 
case $(1)$: $\operatorname{Ham}(\bm{\sigma}, \bm{\sigma}')$ rows in $\mathbf{A}$ and $\mathbf{A}'$ have elements generated from different distributions, and case $(2)$: $n - \operatorname{Ham}(\bm{\sigma}, \bm{\sigma}')$ rows have elements from the same distribution. \\

\indent Case $(1)$: In this case, the probability the corresponding elements in the two matrices $\mathbf{A}$ and $\mathbf{A}'$ are different is $\bar{q} =1 - (q \cdot p + p \cdot q) = 1 - 2 pq$. \\ 

\indent Case $(2)$:  In this case, the probability the corresponding elements in the two matrices $\mathbf{A}$ and $\mathbf{A}'$ are same is  $\bar{p} =1 - (p^{2} + q^{2})$.

Guided by these insights, we are ready to prove our general lower bound.

\begin{theorem}[Necessary Condition]\label{thm:converse}
Define $\Delta
\triangleq
e^{2\epsilon} 
\;+\; 
\bigl(1 - e^{2\epsilon}\bigr)\bigl(p^2 + q^2\bigr)
\;-\; 1$. Suppose there exists an $\epsilon$-edge DP mechanism such that, for any ground truth labeling vector 
$\bm{\sigma}^{*} $ and for $G \sim \operatorname{SBM}(\bm{\sigma}^{*}, n, p, q)$, 
the mechanism outputs 
$\hat{\bm{\sigma}} $ satisfying the 
$(\beta,\eta)$-accurate recovery condition 
\eqref{eqn:beta_eta_accurate_definition}.
Then, a necessary condition is that $n$ must satisfy
\begin{align*}
    n \geq \frac{\beta A + \sqrt{\beta^2 A^2 + 8 \left( 1 - 8 \beta \right) \Delta B}}{8 \beta \left( 1 - 8 \beta \right) \Delta}.
\end{align*}
where $A = \log \left( \frac{1}{8 e^\beta} \right)$ and $B = \log \left( \frac{1}{\eta} \right)$.

\begin{proof}

We next focus in calculating the term $M_{\operatorname{Ham}(\mathbf{A}, \mathbf{A}')}(2 \epsilon) \triangleq \mathds{E}_{\mathbf{A}, \mathbf{A}' \sim \Pi(\mathbf{A}, \mathbf{A}')} \left[ e^{2 \epsilon \operatorname{Ham}(\mathbf{A}, \mathbf{A}')} \right] $ with the following set of steps: 
\begin{align*}
    & M_{\operatorname{Ham}(\mathbf{A}, \mathbf{A}')}(2 \epsilon) \\ = & \left(M_{\operatorname{Ham}(\mathbf{A}, \mathbf{A}'): \text{same dist.}}(2 \epsilon) \right)^{(n - \operatorname{Ham}(\bm{\sigma}, \bm{\sigma}'))} \\
    & \times  \left(M_{\operatorname{Ham}(\mathbf{A}, \mathbf{A}'): \text{different dist.}}(2 \epsilon) \right)^{\operatorname{Ham}(\bm{\sigma}, \bm{\sigma}')}, 
\end{align*}
where, 
\begin{align*}
    M_{\operatorname{Ham}(\mathbf{A}, \mathbf{A}'): \text{same dist.}}(2 \epsilon) & = e^{2 \epsilon} \bar{p} + (1 - \bar{p}), \\
    M_{\operatorname{Ham}(\mathbf{A}, \mathbf{A}'): \text{different dist.}}(2 \epsilon) & =  e^{2 \epsilon} \bar{q} + (1 - \bar{q}).
\end{align*}
We then can readily show that,
\begin{equation}     \label{eqn:MGF_expression_upper_bound}
\begin{aligned}
     & M_{\operatorname{Ham}(\mathbf{A}, \mathbf{A}')}(2 \epsilon)  \\ 
      & \leq   \left(M_{\operatorname{Ham}(\mathbf{A}, \mathbf{A}'): \text{same dist.}}(2 \epsilon) \right)^{(n - \operatorname{Ham}(\bm{\sigma}, \bm{\sigma}')) \cdot \operatorname{Ham}(\bm{\sigma}, \bm{\sigma}')} \nonumber \\ 
     & =  ( e^{2\epsilon} + (1 - e^{2 \epsilon} ) (p^{2} + q^{2}))^{(n - \operatorname{Ham}(\bm{\sigma}, \bm{\sigma}')) \cdot \operatorname{Ham}(\bm{\sigma}, \bm{\sigma}')}.
\end{aligned}
\end{equation}
Plugging \eqref{eqn:MGF_expression_upper_bound} in \eqref{eqn:privacy_utility_condition} yields the following:
\begin{align*}
    \operatorname{Pr}(\hat{\bm{\sigma}}(\mathbf{A}') \in \mathcal{S}_{i}) & \geq \frac{(1-\eta)^{2}}{ M_{\operatorname{Ham}(\mathbf{A}, \mathbf{A}')}(2 \epsilon)}.
\end{align*}
Finally, we lower bound the packing number $m$ with respect to $\operatorname{err} \operatorname{rate}$. 
Building upon the framework established in \cite{chen2023private}, we can readily demonstrate that
\begin{align*}
    m \geq \frac{1}{2} \cdot \frac{| \mathcal{B}_{\operatorname{Ham}} (\bm{\sigma}^{*}, 4\beta n) |}{| \mathcal{B}_{\operatorname{Ham}} (\bm{\sigma}^{*}, 2 \beta n) |},
\end{align*}
where $\mathcal{B}_{\operatorname{Ham}} (\bm{\sigma}^{*}, t \beta n) = \{\bm{\sigma} \in \{\pm 1\}^{n}: \operatorname{Ham}(\bm{\sigma}, \bm{\sigma}^{*}) \leq \beta \} $ and $| \mathcal{B}_{\operatorname{Ham}} (\bm{\sigma}^{*}, t \beta n) |$ is the number of vectors within the Hamming distance of $t \beta n$ from $\bm{\sigma}^{*}$ for $t > 0$. It includes all vectors that can be obtained by flipping any $t \beta n$ elements of $\bm{\sigma}^{*}$. Thus, we can further lower bound $m$ as 
\begin{align*}
    m \geq \frac{1}{2} \cdot \frac{{n \choose 4 \beta n} }{{n \choose 2 \beta n}} \geq \frac{1}{2} \cdot \left(\frac{1}{8 e \beta} \right)^{2 \beta n}.
\end{align*}
Recall that, we have
\begin{align}
   (m-1) \cdot \frac{(1-\eta)^{2}}{ M_{\operatorname{Ham}(\mathbf{A}, \mathbf{A}')}(2 \epsilon)} \leq \eta. \label{eqn:lower_bound_bad_events}
\end{align}
Taking the logarithm for both sides of \eqref{eqn:lower_bound_bad_events}, we have
\begin{equation*}
\begin{aligned} 
   & 8 \beta n (n - 8 \beta n) \cdot \log( e^{2\epsilon} + (1 - e^{2 \epsilon} ) (p^{2} + q^{2})) \\ 
   & \geq  2 \beta n \cdot \log \left(\frac{1}{8 e \beta} \right) + \log \left( \frac{1}{\eta}\right) 
\end{aligned}
\end{equation*}
which implies
\begin{equation*}
\begin{aligned} 
   & 8 \beta n (n - 8 \beta n) \cdot \log( e^{2\epsilon} + (1 - e^{2 \epsilon} ) (p^{2} + q^{2})) \\ \geq & 2 \beta n \cdot \log \left(\frac{1}{8 e \beta} \right) + \log \left( \frac{1}{\eta}\right)  
\end{aligned}
\end{equation*}
and further
\begin{equation*}
\begin{aligned} 
      \Rightarrow & \log( e^{2\epsilon} + (1 - e^{2 \epsilon} ) (p^{2} + q^{2})) \\ \geq & \frac{  \log \left(\frac{1}{8 e \beta} \right)}{4  (n - 8 \beta n)} + \frac{\log \left( \frac{1}{\eta}\right)}{8 \beta n (n - 8 \beta n) } 
\end{aligned}
\end{equation*}
and
\begin{equation} \label{eq:condition_n1}
\begin{aligned} 
 & e^{2\epsilon} + (1 - e^{2 \epsilon} ) (p^{2} + q^{2}) -1 \\ \geq & \frac{  \log \left(\frac{1}{8 e \beta} \right)}{4  (n - 8 \beta n)} + \frac{\log \left( \frac{1}{\eta}\right)}{8 \beta n (n - 8 \beta n) }  \\ 
\end{aligned}
\end{equation}
Solving the last inequality of Eq.~\eqref{eq:condition_n1} for $n$, we get the lower bound
\begin{align*}
    n \geq \frac{\beta A + \sqrt{\beta^2 A^2 + 8 \left( 1 - 8 \beta \right) \Delta B}}{8 \beta \left( 1 - 8 \beta \right) \Delta}.
\end{align*}
where $A = \log \left( \frac{1}{8 e^\beta} \right)$ and $B = \log \left( \frac{1}{\eta} \right)$.

\end{proof}
\end{theorem}

We next present our three private spectral-based algorithms and outline their respective accuracy guarantees. The privacy analysis of the methods is based on standard techniques of DP \cite{dwork2014algorithmic}.


\subsection{Graph Perturbation-Based Mechanism}

We next provide a novel analysis for the spectral method applied on the randomized response released edge-DP adjacency matrix.

\begin{definition}[Randomized Response Perturbation Mechanism]
\label{def:graph-perturbation}
Let \(\mathbf{A}\) be the adjacency matrix of the graph. Under Warner's randomized response \cite{warner1965randomized} with a uniform perturbation parameter \(\mu = 1/(e^{\epsilon} + 1)\), the perturbed adjacency matrix is given by
\begin{equation*}
\hat{\mathbf{A}} = \mathbf{A} + \mathbf{E},
\end{equation*}
where \(\mathbf{E}\) is a symmetric perturbation matrix with i.i.d.\ entries:
\begin{equation} \label{eq:E_def}
E_{ij} =
\begin{cases}
    0, & \text{with probability } 1-\mu, \\
    1 - 2A_{ij}, & \text{with probability } \mu.
\end{cases}
\end{equation}
\end{definition}
The noise matrix \(\mathbf{E}\) defined in Eqn.~\eqref{eq:E_def} clearly satisfies
\begin{equation*}
\begin{aligned} 
\mathds{E}[E_{ij}] &= \mu\,\bigl(1 - 2A_{ij}\bigr), \\
\operatorname{var}(E_{ij}) &= \mu\,(1-\mu)\,\bigl(1 - 2A_{ij}\bigr)^2.
\end{aligned}
\end{equation*}

The analysis is based on decomposing the random release of the adjacency matrix to a deterministic and random terms. We first state some auxilisry results needed for the main result.

\subsubsection{Auxiliary Results for Graph Perturbation Mechanism}

The analysis of the graph perturbation mechanism is based on a decomposition of the error into deterministic and random terms and to this end we first need high-probability bounds for the terms $\|\mathbf{L} - \mathds{E}[\mathbf{L}]\|$ and $\|\hat{\mathbf{L}}-\mathbf{L}\|$.

\begin{lemma} [Concentration of Laplacian Matrices \cite{le2017concentration}] Let $\mathbf{L}$ be a Laplacian whose elements are drawn from a nonhomogeneous Erdös-Rényi model, where each edge $(i,j)$ is generated independently with probability $p_{ij}$. Then, the following holds true with probability at least $1 - \eta$,
\begin{align*}
&    \|\mathbf{L} - \mathds{E}[\mathbf{L}]\| \\ \leq  & C_{\ref{lemma:concentration_laplacian_matrices}} \left( \sqrt{n \max_{(i,j)} p_{ij} \log(n/\eta)} + \log(n/\eta) \right),
\end{align*}
where $C_{\ref{lemma:concentration_laplacian_matrices}}$ is a universal constant, and $\eta \geq n^{-10}$.
\label{lemma:concentration_laplacian_matrices}
\end{lemma}

Next, we introduce two essential lemmas that underpin our main results. 


\begin{lemma} [Concentration of Perturbed Laplacian Matrices via Randomized Response Mechanism] Given a Laplacian matrix $\mathbf{L}$ and its perturbed version $\hat{\mathbf{L}}$ via a randomized response mechanism. The following holds true for a universal constant $C_{\ref{lemma:concentration_perturbation_laplacian_matrices}}$:
    \begin{align*}
    \|\hat{\mathbf{L}}-\mathbf{L}\| \leq  C_{\ref{lemma:concentration_perturbation_laplacian_matrices}}  \sqrt{(\sum_{i<j} \mu_{ij}(1-\mu_{ij})) \cdot \log(n/\eta)}
\end{align*}
with probability at least $1-\eta$.
\label{lemma:concentration_perturbation_laplacian_matrices}
\end{lemma}

\begin{proof}
    
 We are now interested in upper bounding the operator norm (spectral norm) \(\|\hat{\mathbf{L}} - \mathbf{L}\|\).

\indent To apply standard matrix concentration inequalities, define the centered random variables:
\begin{equation*}
\begin{aligned}
X_{ij} &= (\hat{A}_{ij}-A_{ij}) - \mathds{E}[\hat{A}_{ij}-A_{ij}] \\ 
&= (\hat{A}_{ij}-A_{ij}) - \mu_{ij}(1-2A_{ij}).
\end{aligned}
\end{equation*}
Now, $X_{ij}$ are independent (for distinct edges) zero-mean bounded random variables with $|X_{ij}|\leq 1$.

The matrix of interest is:
\begin{equation*}
\begin{aligned}
\hat{\mathbf{A}}- \mathbf{A} = \sum_{i<j} (\hat{A}_{ij}-A_{ij})(\mathbf{E}_{ij}+\mathbf{E}_{ji}),
\end{aligned}
\end{equation*}
where $\mathbf{E}_{ij}$ is the matrix unit with a $1$ in position $(i,j)$ and $0$ elsewhere.

Inserting the centered version
\begin{equation*}
\begin{aligned}
\hat{\mathbf{A}}-\mathbf{A} = \sum_{i<j} [{X}_{ij} + \mu_{ij}(1-2A_{ij})](\mathbf{E}_{ij}+\mathbf{E}_{ji})
\end{aligned}
\end{equation*}
we have
\begin{equation*}
\begin{aligned}
\hat{\mathbf{A}}-\mathbf{A} &= \underbrace{\sum_{i<j} X_{ij}(\mathbf{E}_{ij}+ \mathbf{E}_{ji})}_{\mathbf{Z}} \\ & + \underbrace{\sum_{i<j} \mu_{ij}(1-2A_{ij})(\mathbf{E}_{ij}+\mathbf{E}_{ji})}_{\mathbf{M}}.
\end{aligned}
\end{equation*}

Here, $\mathbf{Z}$ is a zero-mean random symmetric matrix, and $\mathbf{M}$ is a deterministic offset matrix.

Next, consider the degree matrix changes. Since $\hat{D}_{ii} = \sum_j \hat{A}_{ij}$,
\begin{equation*}
\begin{aligned}
\hat{\mathbf{D}}-\mathbf{D} = \mathrm{diag}\left(\sum_j (\hat{A}_{ij}-A_{ij})\right).
\end{aligned}
\end{equation*}
This can also be expressed as a sum of independent random vectors along the diagonal. By a similar reasoning, the change in degree matrix can also be decomposed into a zero-mean part plus a deterministic part.

Overall, we have:
\begin{equation*}
\begin{aligned}
\mathbf{E} &= \hat{\mathbf{L}}-\mathbf{L} = (\hat{\mathbf{D}}-\mathbf{D}) - (\hat{\mathbf{A}}-\mathbf{A}) \\ 
&= (\hat{\mathbf{D}}-\mathbf{D}) - \mathbf{Z} - \mathbf{M}.
\end{aligned}
\end{equation*}
We will focus on the zero-mean part to apply a matrix Bernstein-type inequality. In our case, the random variables $X_{ij}$ correspond to modifications of single entries. The matrices we sum over are rank-2 updates (like $(\mathbf{E}_{ij}+\mathbf{E}_{ji})$ and their diagonal adjustments). Each such update has operator norm at most $2$. 

The variance parameter $\sigma^2$ depends on sums of variances $\mu_{ij}(1-\mu_{ij})$. Summing over all edges, we get:
\begin{equation*}
\begin{aligned}
\sigma^2 = O\left(\sum_{i<j} \mu_{ij}(1-\mu_{ij})\right).
\end{aligned}
\end{equation*}
In a sparse or moderately dense regime, this is often $O(n)$, but depends on the distribution of $p_{ij}$.

Thus, with high probability,
\begin{equation*}
\begin{aligned}
\|\mathbf{Z}\|_2 &= O(\sqrt{\sigma^2 \log n}) \\
&= O\left(\sqrt{(\sum_{i<j} \mu_{ij}(1-\mu_{ij})) \log n}\right).
\end{aligned}
\end{equation*}
Recall $\mathbf{E} = \hat{\mathbf{L}}-\mathbf{L} = (\hat{\mathbf{D}}-\mathbf{D}) - \mathbf{Z} - \mathbf{M}$. A similar argument applies to $(\hat{\mathbf{D}}-\mathbf{D})$, which also can be viewed as a sum of independent diagonal updates. The matrix Bernstein or Vector Bernstein inequality can handle these diagonal terms similarly.

The deterministic part $\mathbf{M}$ is known and can be bounded separately. The main random fluctuation is in $\mathbf{Z}$ and $(\hat{\mathbf{D}}-\mathbf{D})$. Combining these results, we obtain a concentration inequality of the form:
\begin{equation*}
\begin{aligned}
\mathbb{P}\left(\|\hat{\mathbf{L}}- \mathbf{L}\| \ge t\right) \leq 2n \exp\left(\frac{-c t^2}{\sum_{i<j} \mu_{ij}(1-\mu_{ij}) + t}\right),
\end{aligned}
\end{equation*}
for some absolute constant $c$. For $t$ large enough, this yields a bound like:
\begin{equation*}
\begin{aligned}
\|\hat{\mathbf{L}}-\mathbf{L}\| \le C \sqrt{\left(\sum_{i<j} \mu_{ij}(1-\mu_{ij})\right)\log(n/\eta)}
\end{aligned}
\end{equation*}
with probability at least $1-\eta$.
\end{proof}

\subsubsection{Main Result for Graph Perturbation Mechanism}

We are now ready to state the main results for the graph perturbation mechanism.

\begin{theorem}[Distance to Ground Truth Labels]
\label{theorem:distance_bound}
Let $\mathbf{u}_{2}$ and $\hat{\mathbf{u}}_{2}$ be the second eigenvectors of the unperturbed Laplacian matrix $\mathbf{L}$ and the privatized Laplacian matrix $\hat{\mathbf{L}}$ obtained via Warner's randomized response, respectively. Then, with probability at least $1 - 3 \eta$, we have
\begin{equation} \label{eqn_first_upper_bound}
\begin{aligned} 
& \min_{s \in \{\pm 1\}} \| \hat{\mathbf{u}}_2 - s \mathbf{u}_2 \|_2  \nonumber \\
& \leq    \frac{4 \sqrt{2}}{n(p-q)} \cdot \bigg( q n \;+\; \sqrt{8 \mu (1 - \mu) n \log \left(2/\eta \right)} \\
& \quad \quad \quad \quad\quad \quad \;+\; \dfrac{4}{3 \sqrt{n}} \log \left( 2 /\eta \right) \bigg).
\end{aligned}
\end{equation}
\begin{proof}

The proof is based on a \emph{brief perturbation analysis} of the difference 
\(\Delta \mathbf{L} = \hat{\mathbf{L}} - \mathbf{L}\). 
Starting with the definition,
\[
\Delta \mathbf{L} \,\mathbf{u}_2 
\;=\;
(\hat{\mathbf{L}} - \mathbf{L}) \,\mathbf{u}_2 
\;=\;
(\Delta \mathbf{D}) \,\mathbf{u}_2 \;-\; \mathbf{E} \,\mathbf{u}_2,
\]
observe that \(\Delta \mathbf{D}\) is diagonal, and hence its action on 
\(\mathbf{u}_2\) can be recast in terms of \(\mathbf{E}\).  In particular, 
we derive the component-wise expression
\[
(\Delta \mathbf{L} \,\mathbf{u}_2)_i 
\;=\;
\sum_{j=1}^n
E_{ij}\,\bigl(u_{2i} - u_{2j}\bigr),
\]
thereby yielding
\[
\Delta \mathbf{L} \,\mathbf{u}_2
\;=\;
\sum_{i=1}^n \sum_{j=1}^n
E_{ij}\,\bigl(u_{2i} - u_{2j}\bigr)\,\mathbf{e}_i,
\]
where \( \mathbf{e}_i \) is the \( i \)-th standard basis vector.

Next, we decompose \(E_{ij}\) into its mean and zero-mean parts:
\[
E_{ij}
\;=\;
\mu_{ij}\,\bigl(1 - 2A_{ij}\bigr) \;+\; \tilde{E}_{ij},
\quad
\text{where}
\quad
\mathds{E}\bigl[\tilde{E}_{ij}\bigr] \;=\; 0.
\]
Define the random vectors
\[
\mathbf{x}_{ij}
\;=\;
\tilde{E}_{ij}\,\bigl(u_{2j} - u_{2i}\bigr)\,\bigl(\mathbf{e}_i - \mathbf{e}_j\bigr),
\]
so that the difference of Laplacians acting on \(\mathbf{u}_2\) becomes
\[
\Delta \mathbf{L} \,\mathbf{u}_2 
\;=\;
\sum_{i<j}
\mathbf{x}_{ij}
\;+\;
\mathbf{d},
\]
where \(\mathbf{d}\) is a deterministic vector resulting from the means 
\(\mu_{ij} \,(1 - 2 A_{ij})\).

To control the magnitude of the random sum 
\(\sum_{i<j}\mathbf{x}_{ij}\), we use the vector Bernstein Inequality \cite{mitzenmacher2017probability}.  
Specifically, two crucial quantities must be bounded:

\emph{(i)}~The norm of each summand \(\mathbf{x}_{ij}\).  
Since \(\lvert u_{2i} - u_{2j} \rvert \leq D_u\) and 
\(\tilde{E}_{ij}\) takes values in a bounded set, each summand obeys
\[
\|\mathbf{x}_{ij}\|_2
\;\le\;
2 \,\lvert \tilde{E}_{ij}\rvert \,D_u
\;\le\;
2\,D_u
\;\triangleq\;
L_{ij}.
\]
Hence \(L_{ij} \leq 2D_u\). It is worth highlighting that we normalize the labeling vectors by a $1/\sqrt{n}$ factor. In this case, we have $D_{u} = 2/\sqrt{n}$.

\emph{(ii)}~The sum of variances.
We compute
\[
\sigma^2
\;=\;
\sum_{i<j}
\mathbf{E}\bigl[\|\mathbf{x}_{ij}\|_2^2\bigr]
\;=\;
4 \sum_{i<j}
\mu_{ij}\,(1-\mu_{ij})\,(u_{2i} - u_{2j})^2.
\]

By the vector Bernstein Inequality, for any \(t > 0\),
\[
\operatorname{Pr}
\Bigl(
\Bigl\|
\sum_{i<j}\mathbf{x}_{ij}
\Bigr\|_2
\,\ge\, t
\Bigr)
\;\le\;
2\,\exp \Bigl(\frac{-t^2/2}{\,\sigma^2 + (L\,t)/3}
\Bigr),
\]
where \(L = \max_{i<j} L_{ij}\).  Solving in terms of a confidence parameter
\(\eta\), we obtain that with probability at least \(1-\eta\),
\[
\Bigl\|
\sum_{i<j}\mathbf{x}_{ij}
\Bigr\|_2
\;\le\;
\sqrt{2\,\sigma^2 \,\log \bigl({2}/{\eta}\bigr)}
\;+\;
\frac{L}{3}\,\log\bigl({2}/{\eta}\bigr).
\]
Since \(\Delta \mathbf{L}\,\mathbf{u}_2\) is this sum plus the deterministic term
\(\mathbf{d}\), we conclude:
\[
\bigl\|
\Delta \mathbf{L} \,\mathbf{u}_2
\bigr\|_2
\;\le\;
\|\mathbf{d}\|_2
\;+\;
\sqrt{2\,\sigma^2 \,\log \bigl({2}/{\eta}\bigr)}
\;+\;
\frac{L}{3}\,\log \bigl({2}/{\eta}\bigr)
\]
with probability at least \(1 - \eta\).

Substituting back into the Generalized Davis-Kahan Theorem \cite{davis1970rotation}, we obtain:
\begin{equation*}
\begin{aligned}
& \| \mathbf{u}_2 - \hat{\mathbf{u}}_2 \|_2  \\ & \leq \sqrt{2} \cdot \frac{ \| \mathbf{d} \|_2 + \sqrt{ 2 \sigma^2 \log \left( {2}/{\eta} \right) } + \dfrac{L}{3} \log \left( {2}/{\eta} \right) }{ |\hat{\lambda}_3 - \lambda_2| }.
\end{aligned}
\end{equation*}
By Weyl's inequality \cite{horn2012matrix}, we can readily show the following lower bound on $ \hat{\lambda}_{3} - \lambda_{2} $:
\begin{equation}
\begin{aligned}
    \hat{\lambda}_3 - \lambda_2 \geq \frac{n (p -q)}{2}  - 2 \|\mathbf{L} - \mathds{E}[\mathbf{L}]\| - \|\hat{\mathbf{L}}-\mathbf{L}\|. \label{eqn:spectral_grap}
\end{aligned}
\end{equation}
Continuing from the provided inequalities and incorporating the concentration results from Lemma \ref{lemma:concentration_laplacian_matrices} and Lemma \ref{lemma:concentration_perturbation_laplacian_matrices}, we 
arrive at the claim.

    
\end{proof}
\end{theorem}

The above result holds under the following separation condition between $p$ and $q$. Specifically, for a universal constant $C = 4 \cdot \max(2 C_{\ref{lemma:concentration_laplacian_matrices}}, C_{\ref{lemma:concentration_perturbation_laplacian_matrices}})$, with $C_{\ref{lemma:concentration_laplacian_matrices}}$ and $C_{\ref{lemma:concentration_perturbation_laplacian_matrices}}$ given in Lemmas~\ref{lemma:concentration_laplacian_matrices} and~\ref{lemma:concentration_perturbation_laplacian_matrices}, respectively, it is required that
\begin{align*}
& n (p-q)  \;\geq\; C  \Bigl[\mathcal{T}
\;+\; \frac{n}{\sqrt{2}}\, \sqrt{\mu (1-\mu)}\, \sqrt{\log(n/\eta)}\Bigr],
\end{align*}
where $\mathcal{T} \triangleq \sqrt{n p \,\log(n/\eta)} \;+\; \log (n/\eta)$.

This separation condition ensures that the difference $p-q$ is large enough relative to $n$ and the parameters $\mu, \eta$ so that the upper bound in Eqn. \eqref{eqn_first_upper_bound} on the distance to the ground truth labels (denoted as $C_{1}(\epsilon, \eta)$) is meaningful and captures the success of the privatized spectral method.

\begin{figure}[t]
	\centering
    {\includegraphics[width=1.0\columnwidth]{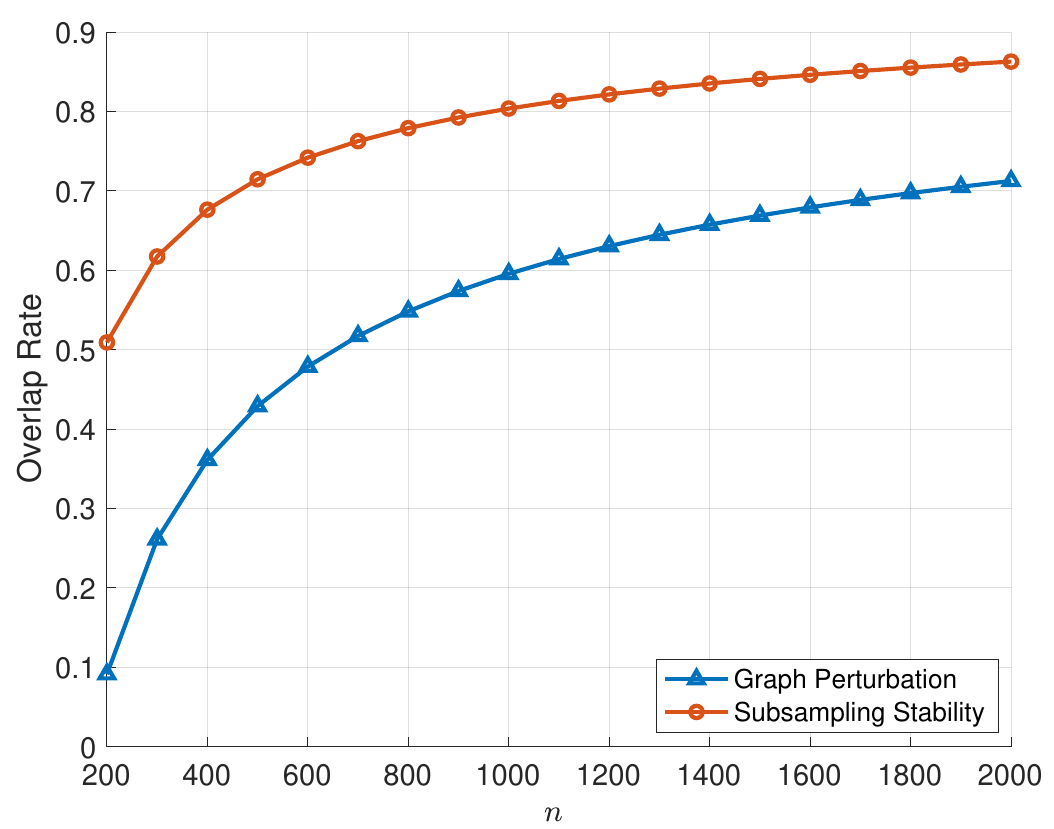}}
    \caption{\small{Overlap rate vs \( n \), for \( \epsilon = 1 \), \( p = 0.25 \), \( q = 0.0025 \), and \( \delta = 10^{-6} \).  A fair comparison is ensured by setting the same total failure probability \( \delta_{\text{failure}} = 0.01\) for both mechanisms. 
        For the Graph-perturbation mechanism, the confidence level \( \eta \) is directly set to \( \delta_{\text{failure}} \). 
        For the  Subsampling stability mechanism, \( \eta \) is  adjusted to satisfy the condition \( 3m\eta = \delta_{\text{failure}} \). }}
    \label{fig:fig0}
\end{figure}

 We next translate the norm difference bound into a statement about the overlap (fraction of correctly identified labels). This connection allows us to interpret the effects of the privacy mechanism directly in terms of the clustering accuracy.

\begin{lemma}[Overlap Rate for Graph Perturbation Mechanism]
\label{lem:lower_bound_overlap}
Consider the graph perturbation mechanism, which satisfies $\epsilon$-edge DP. 
Define the overlap rate between the ground truth labels $\bm{\sigma}^{*}$ and the estimated labels $\hat{\bm{\sigma}}(G)$ obtained via the private spectral method as
\begin{align*} 
\operatorname{overlap\ rate}\bigl(\hat{\bm{\sigma}}(G), \bm{\sigma}^{*}\bigr)
\;\triangleq\;
1 \;-\;
\operatorname{err\ rate}\bigl(\hat{\bm{\sigma}}(G), \bm{\sigma}^{*}\bigr).
\end{align*}
Under the conditions of Theorem~\ref{theorem:distance_bound}, we have
\begin{align*}
\operatorname{overlap\ rate}\bigl(\hat{\bm{\sigma}}(G), \bm{\sigma}^{*}\bigr)
\;\ge\;
1 \;-\; \frac{C_{\ref{theorem:distance_bound}}(\epsilon, \eta)}{8}
\end{align*}
with probability at least \(1 - 3 \eta\), where \(C_{\ref{theorem:distance_bound}}(\epsilon, \eta)\) is determined by the right hand side of the inequality of Theorem~\ref{theorem:distance_bound}.
\end{lemma}

 \subsection{Subsampling Stability-Based Mechanism}

 The key idea in this algorithm is to create $m$ \textit{correlated} subgraphs $\{G_{1},G_{2},\ldots,G_{m}\}$ of the original graph ${G}$ where each subgraph $G_{\ell}$ is generated by randomly subsampling with replacement of the edges in ${G}$ with probability $q_{s}$. We then apply our \textit{non-private} spectral method $\hat{\bm{\sigma}}(G_{\ell})$  on each subgraph $G_{\ell}$. 
 The labeling vectors $\hat{\bm{\sigma}}(G_{\ell})$ are then represented on a histogram. {{Now, define $\text{count}{(\bm{\sigma})} \triangleq | \{ k \in[m]: \hat{\bm{\sigma}}(G_{k}) =  \bm{\sigma}\} | $.}} As shown in \cite{dwork2014algorithmic}, the stability of the histogram is proportional to the difference between the most frequent bin (i.e., the mode) and the second most frequent bin.  In other words, the most frequent outcome of the histogram agrees with the outcome of the original graph with high probability under an appropriate choice of the SBM parameters $(p,q)$, the privacy budget $(\epsilon, \delta)$, the edge sampling probability $q_{s}$, and the number of subsampled weighted graphs $m$. For further details on the mechanism and stability of community detection algorithms over SBMs, please refer to our previous work in \cite{seif2024differentially}. We summarize the mechanism in Algorithm \ref{algo:vanilla_subsampling_stability}.

\begin{algorithm}[t]
  \caption{Subsampling Stability Mechanism}
  \label{algo:vanilla_subsampling_stability}
  \begin{algorithmic}[1]
     \STATE \textbf{Input:} Graph \( G = (\mathcal{V}, E) \), privacy budget \( \epsilon, \delta \), graph structure properties.
     \STATE \textbf{Output:} Private labelling vector \( \hat{\bm{\sigma}} \).
     \STATE Compute the base sampling probability \( q_s \leftarrow \min (1, \epsilon / (32 \log(n)) )\).
     \STATE Compute \( m \leftarrow \lceil \log(n/\delta) / q_s^2 \rceil \).
     \STATE Subsample \( m \) subgraphs \( \{G_1, G_2, \ldots, G_m\} \) using \( q_s \).
     \STATE Compute the label vectors \( \bar{\bm{\sigma}} = (\hat{\bm{\sigma}}(G_1), \ldots, \hat{\bm{\sigma}}(G_m)) \).
     \STATE Aggregate the label vectors via majority voting and compute the stability score:
     \[
     \hat{d} \leftarrow \frac{\text{count}_{(1)} - \text{count}_{(2)}}{4m q_s} - 1.
     \]
     \STATE Add Laplace noise to ensure privacy:
     \[
     \tilde{d} \leftarrow \hat{d} + \operatorname{Lap}(0, 1/\epsilon).
     \]
     \IF{\( \tilde{d} > \log(1/\delta)/\epsilon \)}
     \STATE Output \( \hat{\bm{\sigma}}_{\text{final}} = \text{mode}(\bar{\bm{\sigma}}) \).
     \ELSE
     \STATE Output \( \perp \) (a random label vector).
     \ENDIF
  \end{algorithmic}
\end{algorithm}

\subsubsection{Auxiliary Result for the Subsampling Stability Mechanism}

The proof of the main theorem for the Subsampling Stability Mechanism is based on the same perturbation analysis as the analysis of the graph perturbation mechanism presented in the previous subsection. For this analysis, we will need the following concentration bound for the Laplacian matrix released by the subsampling stability mechanism.

\begin{lemma}[Concentration of Subsampled Laplacian Matrices]
Consider an undirected graph with \( n \) nodes and edge set \( E \), represented by the Laplacian matrix \( \mathbf{L} \). Let \( \hat{\mathbf{L}} \) be the Laplacian matrix of a subsampled graph, obtained by independently including each edge \( (i, j) \in E \) with probability 
\[
q_{s,ij} = q_s \cdot (1 - r(i,j)),
\]
where \( q_s \) is a base sampling probability, \( r(i,j) \in [0, 1] \) is a removal probability determined by the edge-specific subsampling mechanism. The following concentration bound holds for a universal constant \( C_{\ref{lemma:concentration_subsampling_laplacian_matrices_subsampling}} \):
\[
\|\hat{\mathbf{L}} - \mathbf{L}\| 
\;\leq\; 
C_{\ref{lemma:concentration_subsampling_laplacian_matrices_subsampling}} 
\sqrt{( \sum_{i < j} q_{s,ij} (1 - q_{s,ij}) ) \cdot \log(n / \eta)},
\]
with probability at least \( 1 - \eta \).
\label{lemma:concentration_subsampling_laplacian_matrices_subsampling}
\end{lemma}

\subsubsection{Main Result for the Subsampling Stability Mechanism}

 \begin{theorem}[Distance to Ground Truth Labels]
 \label{theorem:distance_bound_stability}
 Let $\mathbf{u}_{2}$ and $\hat{\mathbf{u}}_{2}$ be the second eigenvectors of the unperturbed Laplacian matrix $\mathbf{L}$ and the Laplacian matrix $\hat{\mathbf{L}}$ of the subsampled graph $G_{\ell}$ obtained via the subsampling mechanism, respectively. Then, with probability at least $1 - 3 \eta$, we have
\begin{equation}
\begin{aligned}
& \min_{s \in \{\pm 1\}} \| \hat{\mathbf{u}}_2 - s \mathbf{u}_2 \|_2 \nonumber \\ 
&\leq 
  \frac{4 \sqrt{2}}{n(p-q)} \cdot \bigg( \|\mathbf{d}\|_2 + \sqrt{2 \sigma^2 \log \left({2}/{\eta}\right)} \\ 
& \quad \quad \quad \quad \quad \quad \quad \quad \quad + \frac{L}{3} \log \left({2}/{\eta}\right) \bigg)
\triangleq  C_{\ref{theorem:distance_bound_stability}}(\epsilon, \eta),
\end{aligned}
\end{equation}
where, \( \mathbf{d} \) is a deterministic part of the perturbation, derived from the adjusted edge sampling mechanism where \( \|\mathbf{d}\|_2 = \frac{2q_s}{\sqrt{n}} \sqrt{\frac{q}{p+q} \cdot |E|} \), \( |E| \) represents the total number of edges in the graph, and \( \sigma^2 = \frac{16}{n} \cdot q_s (1 - q_s) |E_{\text{inter}}| \), where \( |E_{\text{inter}}| \) denotes the number of inter-community edges. This is the variance of the edge sampling noise, with \( q_{s} \) which can be further upper bounded as \( \sigma^2 \leq \frac{4}{n} \cdot \frac{q}{p+q} \cdot |E| \), and \( L = \max_{(i,j) \in E} \|\mathbf{x}_{ij}\|_2 \leq {4}/{\sqrt{n}} \) is the upper bound on the norm of the noise components.
\begin{proof}
The proof is based on the same perturbation analysis as the analysis of the graph perturbation mechanism given in the proof of Thm.~\ref{theorem:distance_bound}, with the bound of Lemma~\ref{lemma:concentration_subsampling_laplacian_matrices_subsampling} used instead for the high-probability bound for the term $\|\hat{\mathbf{L}} - \mathbf{L}\|$.
\end{proof}
\end{theorem}

The above result holds under the following separation condition between $p$ and $q$. Specifically, for a universal constant $C = 4 \cdot \max\left(2 C_{\ref{lemma:concentration_laplacian_matrices}}, \, C_{\ref{lemma:concentration_subsampling_laplacian_matrices_subsampling}} \right)$, with $C_{\ref{lemma:concentration_laplacian_matrices}}$ and $C_{\ref{lemma:concentration_subsampling_laplacian_matrices_subsampling}}$ given in Lemmas~\ref{lemma:concentration_laplacian_matrices} and~\ref{lemma:concentration_subsampling_laplacian_matrices_subsampling}, respectively, it is required that
\begin{align*}
& n (p - q)  \;\geq\; C  \Bigl[\mathcal{T}
\;+\; \frac{n}{\sqrt{2}} \sqrt{q_{s} (1-q_{s}) } \sqrt{\log(n/\eta)}\Bigr],
\end{align*}
where $\mathcal{T} \triangleq \sqrt{n p \,\log(n/\eta)} \;+\; \log \left({n}/{\eta}\right)$.

\begin{remark}[Impact of \( q_s \) on the Distance Bound]
The  edge sampling probability \( q_s \) plays a crucial role in the trade-off between privacy and spectral accuracy. As \( q_s \) decreases:
\begin{itemize}
    \item The deterministic perturbation \( \|\mathbf{d}\|_2 \) and the variance \( \sigma^2 \) decrease, reflecting reduced magnitudes of the subsampling-induced noise.
    \item However, the spectral gap \( \hat{\lambda}_3 - \lambda_2 \) (refer to Eqn. \eqref{eqn:spectral_grap}) also decreases significantly due to the increased instability of the subsampled graph, particularly when critical edges are removed.
\end{itemize}
Here, \( \hat{\lambda}_3 \) is the third smallest eigenvalue of the subsampled graph \( G_\ell \). This reduction in the spectral gap dominates the bound, leading to an overall increase in the upper bound on the distance \( \| \hat{\mathbf{u}}_2 - s \mathbf{u}_2 \|_2 \). 
\end{remark}

\begin{lemma}[Overlap Rate for  Subsampling Stability Mechanism]
\label{lem:lower_bound_overlap_subsampling_stability}
Consider the  subsampling stability-based mechanism, which satisfies $(2 \epsilon, \delta)$-edge DP. The overlap rate between the ground truth labels \( \bm{\sigma}^{*} \) and the final estimated labels \( \hat{\bm{\sigma}}_{\text{final}} \) obtained via majority voting on \( m \) correlated subgraphs \( \{G_1, \ldots, G_m\} \) satisfies:
\[
\operatorname{overlap\,rate} (\hat{\bm{\sigma}}_{\text{final}}, \bm{\sigma}^{*}) 
\geq 1 - \frac{C_{\ref{theorem:distance_bound_stability}}(\epsilon, \eta)}{4} - O\left(\frac{1}{\sqrt{m}}\right),
\]
with probability at least \( 1 - 3 m \eta \), where \( C_{\ref{theorem:distance_bound_stability}}(\epsilon, \eta) \) is the error contribution from the non-private spectral method for individual subsampled Graph \( G_{\ell} \).
\end{lemma}

\textit{Proof Sketch.} For a single subsampled graph \( G_\ell \), the labeling \( \hat{\bm{\sigma}}(G_\ell) \) satisfies: \( \operatorname{overlap\,rate} (\hat{\bm{\sigma}}(G_\ell), \bm{\sigma}^{*}) \geq 1 - C_{\ref{theorem:distance_bound_stability}} \), implying that at most \( C_{\ref{theorem:distance_bound_stability}} n \) nodes are mislabeled in each subgraph. Let \( S_\ell \subseteq [n] \) denote the set of mislabeled nodes for subgraph \( G_\ell \), so \( |S_\ell| \leq C_{\ref{theorem:distance_bound_stability}} n \). The final labeling \( \hat{\bm{\sigma}}_{\text{final}} \) is obtained by taking a majority vote over the \( m \) subgraphs. A node \( i \) is mislabeled in \( \hat{\bm{\sigma}}_{\text{final}} \) if it is mislabeled in more than \( m/2 \) subgraphs.

Let \( X_i \) denote the number of subgraphs in which node \( i \) is mislabeled. Since \( \mathbb{E}[X_i] = m C_{\ref{theorem:distance_bound_stability}
} \), the probability that \( X_i \geq m/2 \) can be bounded using Hoeffding's inequality:
\[
\operatorname{Pr}(X_i \geq {m}/{2}) \leq \exp\left( {-2 \left(\frac{m}{2} - m C_{\ref{theorem:distance_bound_stability}}\right)^2}/{m}\right).
\]

Aggregating the probability of mislabeling across all nodes, the fraction of mislabeled nodes (denoted as $\mathcal{M}$) in the final labeling satisfies that  $\frac{|\mathcal{M}|}{n} \leq 2C_{\ref{theorem:distance_bound_stability}} + O\left({1}/{\sqrt{m}}\right)$,
where \( 2C_{\ref{theorem:distance_bound_stability}} \) accounts for the worst-case overlap between mislabeled nodes across subgraphs.

\begin{table*}[!t]
\centering
\renewcommand{\arraystretch}{1.8}
\caption{Comparison of error bounds and summarization of the practical methods: Graph Perturbation vs. Noisy Power Iteration. }
\begin{tabular}{|c|p{7cm}|p{6cm}|}
\hline
\textbf{Method} & \textbf{Graph Perturbation Mechanism} & \textbf{Noisy Power Iteration} \\
\hline
\textbf{Error in the Second Eigenvector} 
& 
$ \min_{s \in \{\pm 1\}} \| \hat{\mathbf{u}}_2 - s\mathbf{u}_2 \|_2 = O \bigg(
\frac{q}{p-q} + \frac{ 1 }{
e^{\epsilon/2}  (p - q)\sqrt{n} } 
\bigg)$ 
& 
$ \min_{s \in \{\pm 1\}} \| \mathbf{y}_N - s\mathbf{u}_2 \|_2 = O \bigg(
\frac{ \sqrt{\log n} }{
\epsilon (p - q)\sqrt{n} } \bigg)$ \\
\hline
\textbf{Noise Source}
& Randomized graph perturbation
& Added noise in iterative updates \\
\hline
\textbf{Time Complexity}
& $O(n^3)$ with full eigendecomposition \newline $\tilde{O}(n^2)$ (leaving log factors) with power or Lanczos iteration
& Dense SBM: $O\big((n \log n)N\big)$ = $O(n (\log n)^2)$ \newline Sparse SBM: $O(n N)$ = $O(n \log n)$ \\
\hline
\textbf{Space Complexity}
& $O(n^2)$
& Dense SBM: $O(n \log n)$, Sparse SBM: $O(n)$ \\
\hline
\end{tabular}
\label{tab:error_bounds}
\end{table*}

\subsection{Noisy Power Iteration Method}


As motivation, we build on the approach of~\cite{wang2020nearly}, which applies
power iteration to the centered adjacency matrix
\begin{align*}
    \mathbf{B} =\ \mathbf{A} -  \rho\,\mathbf{1}\mathbf{1}^{\!\top},
    \quad
    \rho  = {\mathbf{1}^{\!\top}\mathbf{A}\mathbf{1}}/{n^{2}}. 
    \end{align*}
We tailor this procedure to the differential-privacy setting by replacing the
standard power iteration with the \emph{noisy power method} of~\cite{hardt2014noisy}:
\begin{align*}
\mathbf{x}_t = \mathbf{B} \mathbf{y}_{t-1} + \mathbf{z}_t,  \quad \mathbf{y}_{t-1} = \mathbf{x}_{t-1}/\norm{\mathbf{x}_{t-1}}_2,
\end{align*}
where $\mathbf{z}_t \sim \mathcal{N}(0,C^2 \sigma^2 \mathbf{I}_n)$ and $C$ limits the 2-norm sensitivity of the product $\mathbf{B} \mathbf{y}_{t-1}$ with respect to change of a single edge in the graph. The DP guarantees are then obtained from a composition analysis of the Gaussian mechanism.
To this end, we need to bound the sensitivty of the multiplication $\mathbf{B} \mathbf{y}_{t-1}$ w.r.t. to change of a single edge, as formalized in the following lemma.

\begin{lemma} \label{lem:sensitivity_bound}
Let $\mathbf{A}$ and $\mathbf{A}'$ differ in a single element, and let $\mathbf{y} \in \mathds{R}^n$ with $\norm{\mathbf{y}}_2=1$. Then,
$$
\norm{\mathbf{B} \mathbf{y} - \mathbf{B}'\mathbf{y}}_2  \leq \norm{\mathbf{y}}_\infty + \frac{1}{n}.   
$$
\begin{proof}
Let $\mathbf{A}$ and $\mathbf{A}'$ differ in $(i,j)$th element.
As $\mathbf{B} = \mathbf{A} - \rho \mathbf{1} \mathbf{1}^T$, where $ \rho= \mathbf{1}^T \mathbf{A} \mathbf{1} / n^2$ (and similarly $\mathbf{B}'  = \mathbf{A}' - \rho' \mathbf{1} \mathbf{1}^T$ for $ \rho' = \mathbf{1}^T \mathbf{A}' \mathbf{1} / n^2$), we have that:
\begin{equation*}
\begin{aligned}
    \norm{\mathbf{B} \mathbf{y} - \mathbf{B}'\mathbf{y}}_2  &\leq \norm{\mathbf{A} \mathbf{y} - \mathbf{A}'\mathbf{y}}_2 \\
    & \qquad + \norm{(\mathbf{1}^T \mathbf{A} \mathbf{1} / n^2 - \mathbf{1}^T \mathbf{A}' \mathbf{1} / n^2) \mathbf{1} \mathbf{1}^T \mathbf{y}}_2 \\
    & = \norm{y_j \mathbf{e}_i}_2 + \frac{1}{n^2} \norm{(\mathbf{1}^T (\mathbf{A} - \mathbf{A}') \mathbf{1})  \mathbf{1} \mathbf{1}^T \mathbf{y}} \\ 
   &=  |y_j| + \frac{1}{n^2} \norm{  \mathbf{1} \mathbf{1}^T \mathbf{y}}_2 \\
   &=  |y_j| + \frac{1}{n^2} \norm{  \mathbf{1}}_2 |\mathbf{1}^T \mathbf{y}| \\
    &=  |y_j| + \frac{|\sum_{i=1}^n y_i|}{n^{\frac{3}{2}}} \\
    &\leq  \max_{j} |y_j| + \frac{\sum_{i=1}^n |y_i|}{n^{\frac{3}{2}}} \\
    &=  \norm{\mathbf{y}}_\infty + \frac{\norm{\mathbf{y}}_1}{n^{\frac{3}{2}}} \\
   & \leq   \norm{\mathbf{y}}_\infty + \frac{\sqrt{n}\norm{\mathbf{y}}_2}{n^{\frac{3}{2}}} \\
     & =  \norm{\mathbf{y}}_\infty + \frac{1}{n},
\end{aligned}
\end{equation*}
where \( \mathbf{e}_i \) is the \( i \)-th standard basis vector. In the last inequality we have used the relation $\| \mathbf{y} \|_1 \leq \sqrt{n} \| \mathbf{y} \|_2$ which holds for all $\mathbf{y} \in \mathbb{R}^n$, and in the last equality the assumption $\norm{\mathbf{y}}_2=1$.  This completes the proof of the lemma.
\end{proof}
\end{lemma}

With the above $\mathbf{y}$-adaptive sensitivity bound, we get the noisy power method depicted in Algorithm~\ref{algo:noisy_pm}.



\begin{algorithm}[t]
  \caption{Noisy Power Iteration}
  \label{algo:noisy_pm}
  \begin{algorithmic}[1]
     \STATE \textbf{Input:} Adjacency matrix $\mathbf{A} \in \{0,1\}^{n \times n}$, positive integer $N$
     \STATE \textbf{Output:} Private labelling vector $ \hat{\sigma} $.
     \STATE Set $\rho= \mathbf{1}^T \mathbf{A} \mathbf{1} / n^2$, $ \mathbf{B} = \mathbf{A} - \rho \mathbf{1} \mathbf{1}^T$.
     \STATE Draw $\mathbf{y}_0$ randomly from the unit sphere $\mathbb{B}^{n-1}$.
     \FOR{$t = 1$ to $N$}
        \STATE $\mathbf{x}_t = \mathbf{B} \mathbf{y}_{t-1} + \mathbf{z}_t$, \\ \quad \quad $\mathbf{z}_t \sim \mathcal{N}\left(0, (\norm{\mathbf{y}_{t-1}}_{\infty}+\tfrac{1}{n})^2 \sigma^2 \mathbf{I}_n\right)$.
        \STATE $\mathbf{y}_t = \mathbf{x}_t/\norm{\mathbf{x}_t}_2$.
    \ENDFOR 
    \STATE $ \hat{\bm{\sigma}} $ = $\mathbf{y}_t / |\mathbf{y}_t|$.
  \end{algorithmic}
\end{algorithm}
\begin{remark}
The privacy guarantee of Algorithm~\ref{algo:noisy_pm} follows from the facts that analyzing an $N$-wise composition of Gaussian mechanisms, each with noise ratio $\sigma$, is equivalent to analyzing a Gaussian mechanism with noise ratio $\sigma/\sqrt{N}$ and by applying standard tail bounds for the Gaussian distribution~\cite{dwork2014algorithmic}.  In particular, choosing
\(
\sigma
    = \frac{1}{\epsilon}\sqrt{4N \log \!\bigl(1/\delta\bigr)}
\)
ensures that the entire sequence
\(
\mathbf{x}_{1},\ldots,\mathbf{x}_{N}
\)
is \((\epsilon,\delta)\)-differentially private.  
Note that Step~9 involves only post-processing of data-dependent intermediate
values and therefore incurs no additional privacy cost.    
\end{remark}

\subsubsection{Auxiliary Results for Noisy Power Iteration Method}

The utility analysis of Algorithm~\ref{algo:noisy_pm} follows directly from~\cite[Thm.\;1.3]{hardt2014noisy}.
The general result of~\cite[Thm 1.3]{hardt2014noisy} is stated for a block matrix iteration. By carefully following the proof of~\cite[Thm 1.3]{hardt2014noisy}, we observe that it can be adapted to yield a "$1-\eta$" high-probability bound by replacing~\cite[Lemma A.2]{hardt2014noisy} with the following result, which is specifically tailored to the noisy power vector iteration and applied using the sensitivity bound of Lemma~\ref{lem:sensitivity_bound}.




\begin{lemma} \label{lem:max_G_norms}
Let $\mathbf{u} \in \mathbb{R}^{n}$ be a unit vector, and let $\mathbf{g}_1, \dotsc, \mathbf{g}_L \sim \mathcal{N}(0, \sigma^2 I_n)$ be independent Gaussian vectors. Then, for any $\eta \in (0,1)$, with probability at least $1 - \eta$, we have simultaneously
\begin{equation*}
    \begin{aligned}
&\max_{\ell \in [N]} |\mathbf{u}^\top \mathbf{g}_\ell| \leq \sigma\sqrt{2 \log\left( \frac{2N}{\eta} \right)}
\quad \text{and} \quad \\
&\max_{\ell \in [N]} \|\mathbf{g}_\ell\| \leq \sigma\left( \sqrt{n} + \sqrt{2 \log\left( \frac{2N}{\eta} \right)} \right).
 \end{aligned}
\end{equation*}
\begin{proof}
Since $\mathbf{u}$ is a unit vector, each $\mathbf{u}^\top \mathbf{g}_\ell$ is distributed as a one-dimensional Gaussian $\mathcal{N}(0, \sigma^2)$. Standard Gaussian concentration inequalities imply that for $t \geq 0$,
\begin{equation*}
    \begin{aligned}
        &\operatorname{Pr}\left( |\mathbf{u}^\top \mathbf{g}_\ell| \geq \sigma t \right) \leq 2e^{-t^2/2}
\quad \text{and} \quad \\
&\operatorname{Pr}\left( \|\mathbf{g}_\ell\| \geq \sigma(\sqrt{n} + t) \right) \leq e^{-t^2/2}.
    \end{aligned}
\end{equation*}
Applying a union bound over the $2N$ events and setting $t = \sqrt{2 \log\left( {2N}/{\eta} \right)}$, the claim follows.
\end{proof}
\end{lemma}

Using Lemma~\ref{lem:max_G_norms} in the proof of~\cite[Thm 1.3]{hardt2014noisy} instead of~\cite[Lemma A.2]{hardt2014noisy}, we directly have the following high-probability version of~\cite[Thm 1.3]{hardt2014noisy}.

\begin{lemma} \label{lem:ref_bound}
If we choose $\sigma = \epsilon^{-1} \sqrt{4 N \log(1/\delta)}$, then the Noisy Power Iteration of Algorithm\ref{algo:noisy_pm} with $N$ iterations satisfies $(\epsilon, \delta)$-edge DP.  Moreover, after $ N = O\left( \frac{\lambda_1}{\lambda_1 - \lambda_2} \log n  \right) $ iterations we have with probability at least $1-\eta$ that

\begin{equation} \label{eq:statement1}
\begin{aligned}
    & \left\|  (\mathbf{I} - \mathbf{u}_2 \mathbf{u}_2^T) \mathbf{y}_N \right\|_2 \leq  \\
    &  \frac{\sigma (\max_{t \in [N]} \| \mathbf{y}_t \|_{\infty} + \tfrac{1}{n} )\left( \sqrt{n} + \sqrt{2 \log\left( {2N}/{\eta} \right)} \right)}{\lambda_1 - \lambda_2}  ,
\end{aligned}
\end{equation}
where $\lambda_1 \geq \lambda_2$ denote the two largest eigenvalues of the 
matrix $\mathbf{B}$ and $\mathbf{u}_2$ denotes the eigenvector corresponding to the eigenvalue $\lambda_1$.
\end{lemma}

For lower bounding the spectral gap $\lambda_1 - \lambda_2$ for the shifted matrix $\mathbf{B}$, have the following lemma from~\cite{wang2020nearly}.

\begin{lemma} \label{lem:lemma_spectral_gap}
Let $p$ and $q$ be parametrized as
\begin{equation} \label{eq:p_and_q}
p = \frac{\alpha \log n}{n}, \quad q = \frac{\beta \log n}{n}
\end{equation}
for some constant $\alpha > \beta > 0$. Let $\lambda_1 \geq \lambda_2 \geq \ldots \geq \lambda_n$ be the eigenvalues of the matrix $\mathbf{B} = \mathbf{A} - \rho \mathbf{1} \mathbf{1}^T$, where $ \rho= \mathbf{1}^T \mathbf{A} \mathbf{1} / n^2$. Then, for sufficiently large $n$, for some constants $c_1$ and $c_2$, it holds with probability at least $1 - 2 n^{- \frac{1}{2(\alpha + \beta + 1)}}  - c_2 n^{-3}$ that
\begin{align*}
    \lambda_1 \geq \frac{\alpha - \beta}{3} \log n
\end{align*}
and
\begin{align*}
    |\lambda_i | \leq 2 c_1 \sqrt{\log n}, \quad i=2, \ldots, n.
\end{align*}
\end{lemma}

We directly get the following corollary from Lemma~\ref{lem:lemma_spectral_gap}.


\begin{corollary} \label{cor:lambda12} 
Suppose $p = \frac{\alpha \log n}{n}$ and $q = \frac{\beta \log n}{n}$ for constants $\alpha > \beta > 0$, and let $\lambda_1 \geq \lambda_2 \geq \cdots \geq \lambda_n$ be the eigenvalues of the matrix $\mathbf{B} = \mathbf{A} - \rho \mathbf{E}_n$, with $\rho = \frac{1}{n^2} \mathbf{1}_n^\top \mathbf{A} \mathbf{1}_n$-

Then, for all sufficiently large $n$, with probability at least  
$
1 - 2n^{-\frac{1}{2(\alpha + \beta + 1)}} - c_2 n^{-3},
$  
it holds that  
$$
\frac{1}{\lambda_1 - \lambda_2} \leq \frac{1}{\frac{1}{3}(p - q)n - 2c_1 \sqrt{\log n}},
$$  
where $c_1$ and $c_2$ are the constants from Lemma~\ref{lem:lemma_spectral_gap}.
\begin{proof}
Lemma~\ref{lem:lemma_spectral_gap} directly gives the following lower bound for the spectral gap:
\begin{align*}
\lambda_1 - \lambda_2 
&\geq \frac{\alpha - \beta}{3} \log n - 2c_1 \sqrt{\log n}.
\end{align*}
Taking reciprocals yields:
$$
\frac{1}{\lambda_1 - \lambda_2} \leq \frac{1}{\frac{\alpha - \beta}{3} \log n - 2c_1 \sqrt{\log n}}.
$$
Substituting $p$ andn $q$ into the expression above:
\begin{align*}
\frac{1}{\lambda_1 - \lambda_2} 
&\leq \frac{1}{\frac{1}{3} \cdot \frac{(p - q)n}{\log n} \cdot \log n - 2c_1 \sqrt{\log n}} \\
&= \frac{1}{\frac{1}{3}(p - q)n - 2c_1 \sqrt{\log n}}.
\end{align*}
This completes the proof of Corollary \ref{cor:lambda12}.
\end{proof}
\end{corollary}

\subsubsection{Main Result for Noisy Power Iteration}

We are now ready to prove the main convergence theorem for the noisy power method. 
Notice that we can always bound $\max_{t \in [N]} \| \mathbf{y}_t \|_{\infty}$ by 1 since $y_t$ has unit norm for all $t \in [N]$.

\begin{theorem} \label{thm:main_power_iteration}
If we choose $\sigma = \epsilon^{-1} \sqrt{4 N \log(1/\delta)}$ and $C=\norm{\mathbf{y}}_\infty + \frac{1}{n}$, the Noisy Power Iteration with $N$ iterations satisfies $(\epsilon, \delta)$-DP. Moreover, with $ N = O\left( \frac{\lambda_1}{\lambda_1 - \lambda_2} \log n  \right) $ iterations,
for some constants $c_1$ and $c_2$, it holds 
with probability at least $ 1 - 2 n^{- \frac{1}{2(\alpha + \beta + 1)}}  - c_2 n^{-3} $ that
\begin{equation*} 
\begin{aligned}
\min_{s \in \{\pm 1\}} 
 \| \mathbf{y}_N - s \mathbf{u}_2  \|_2 
 \leq \frac{\sqrt{2} \sigma (1 + \tfrac{1}{n} )\left( \sqrt{n} + \sqrt{2 \log\left( \frac{2N}{\eta} \right)} \right)}{\frac{1}{3}(p - q)n - 2c_1 \sqrt{\log n}}.
\end{aligned}
\end{equation*}
\begin{proof}
By simple linear algebra, we first derive a lower bound for the left-hand side of the inequality~\eqref{eq:statement1}.
Since $\mathbf{u}_2$ and $\mathbf{y}_N$ are of unit norm, we have that
\begin{equation*} 
\begin{aligned}
    \left\| (\mathbf{I} - \mathbf{u}_2 \mathbf{u}_2^T) \mathbf{y}_N \right\|_2^2 
   & =    \| \mathbf{y}_N \|_2^2 - 2 (\mathbf{u}_2^T \mathbf{y}_N)^2 + (\mathbf{u}_2^T \mathbf{y}_N)^2 \\
   & = 1 - (\mathbf{u}_2^T \mathbf{y}_N)^2 
\end{aligned}
\end{equation*}
and furthermore, since $\norm{\mathbf{y}_N}_2=\norm{\mathbf{u}_2}_2=1$,
\begin{equation*}
\begin{aligned}
   \min_{s \in \{\pm 1\}}   \| \mathbf{y}_N - s \mathbf{u}_2  \|_2^2 
    & = \min_{s \in \{\pm 1\}} \left( 2 -  2 s(\mathbf{u}_2^T \mathbf{y}_N) \right)\\
    & =  2 -  2 |\mathbf{u}_2^T \mathbf{y}_N | \\
    & \leq 2 \cdot \left( 1 - (\mathbf{u}_2^T \mathbf{y}_N)^2 \right) \\
    & = 2 \left\| (\mathbf{I} - \mathbf{u}_2 \mathbf{u}_2^T) \mathbf{y}_N \right\|_2^2
\end{aligned}
\end{equation*}
which implies 
\begin{equation} \label{eq:uy2}
\begin{aligned}
   \min_{s \in \{\pm 1\}}   \| \mathbf{y}_N - s \mathbf{u}_2  \|_2 
    \leq  \sqrt{2} \left\| (\mathbf{I} - \mathbf{u}_2 \mathbf{u}_2^T) \mathbf{y}_N \right\|_2.
\end{aligned}
\end{equation}
The claim follows then from the inequality~\eqref{eq:uy2}, Lemma~\ref{lem:ref_bound}  and Cor.~\ref{cor:lambda12}.
\end{proof}
\end{theorem}

\begin{corollary}
Under the assumptions of Thm.~\ref{thm:main_power_iteration}, we have that with probability at least $ 1 - 2 n^{- \frac{1}{2(\alpha + \beta + 1)}}  - c_2 n^{-3} - \eta$,
\begin{equation*}
\begin{aligned}
\min_{s \in \{\pm 1\}} 
 \| \mathbf{y}_N - s \mathbf{u}_2  \|_2  = O \left( \frac{\sqrt{\log 1/\delta}}{\epsilon ( p - q)} \sqrt{\frac{\log n}{n}}  \right).
\end{aligned} 
\end{equation*}
\begin{proof}
We also have that
\begin{equation*}
\begin{aligned}
 \frac{\lambda_1}{\lambda_1 - \lambda_2}   & = 1 + \frac{\lambda_2}{\lambda_1 - \lambda_2} \\
& = 1 + O \left( \frac{\sqrt{\log n}}{(p-q)n}  \right)  = O(1)
\end{aligned}
\end{equation*}
and therefore $N = O\left( \frac{\lambda_1}{\lambda_1 - \lambda_2} \log n \right) = O(\log n)$.

Substituting $\sigma$ and $N$ into eqn.~\eqref{eq:statement1} and neglecting a $\log N = O(\log \log n)$ factor, we have that
with probability at least $ 1 - 2 n^{- \frac{1}{2(\alpha + \beta + 1)}}  - c_2 n^{-3} - \eta$,
\begin{equation*}
\begin{aligned}
\| \mathbf{u}_2 - \mathbf{y}_N \|_2 = O \left( \frac{\sqrt{\log 1/\delta}}{\epsilon ( p - q)} \sqrt{\frac{\log n}{n}}  \right).
\end{aligned}
\end{equation*}  
This completes the proof of the corollary.
\end{proof}
\end{corollary}

\begin{remark}
We remark that, analogous to the graph perturbation-based mechanism, the error bound provided in Theorem~\ref{thm:main_power_iteration} for the noisy power iteration method can be directly translated into a lower bound on the overlap rate using Lemma~\ref{lem:lower_bound_overlap}. We formalize this in the following lemma.
\end{remark} 

\begin{lemma}[Overlap Rate for Noisy Power Iteration]
\label{lem:overlap_npi}
Consider the private spectral method based on noisy power iteration that estimates the second eigenvector $\mathbf{u}_2$ of the unperturbed graph Laplacian. Let $\mathbf{y}_N$ denote the final estimate after $N$ iterations, and suppose that
\[
\min_{s \in \{\pm 1\}} \| \mathbf{y}_N - s \mathbf{u}_2 \|_2 \leq \Delta,
\]
with probability at least \( 1 - \tilde{\eta} \). Here, $\tilde{\eta}$ is taken directly from Theorem~\ref{thm:main_power_iteration}, and $\Delta$ corresponds to the upper bound on the Euclidean distance given in the same theorem. Then, the overlap rate between the estimated labels $\hat{\bm{\sigma}} = \operatorname{sign}(\mathbf{y}_N)$ and the true labels $\bm{\sigma}^*$ satisfies
\[
\operatorname{overlap\ rate}(\hat{\bm{\sigma}}, \bm{\sigma}^{*}) \geq 1 - \frac{\Delta^2}{8},
\]
with probability at least \( 1 - \eta \).
\end{lemma}

\subsubsection{Tighter Privacy Analysis for Noisy Power Iteration}

To further optimize the privacy-utility trade-offs for the noisy power iteration, we consider a tighter privacy analysis via so-called dominating pairs of distributions~\cite{zhu2022optimal}.

The noisy power iteration algorithm with $N$ steps can be see as an adaptive composition of the form
\begin{equation*} 
	\begin{aligned}
\mathcal{M}^{(N)}(G) &= \big( \mathcal{M}_1(G), \mathcal{M}_2(\mathcal{M}_1(G),G), \ldots, \\ 
& \quad \quad \quad \mathcal{M}_N(\mathcal{M}_1(G), \ldots, \mathcal{M}_{N-1}(G),G) \big).
\end{aligned}
\end{equation*}
We obtain accurate $(\epsilon,\delta)$-differential privacy guarantees for adaptive compositions by leveraging dominating pairs of distributions, as introduced in~\cite{zhu2022optimal}. Due to the scaling of the noise in Algorithm~\ref{algo:noisy_pm}, the privacy loss at each step is dominated by that of the Gaussian mechanism with sensitivity 1 and noise standard deviation $\sigma$. Accordingly, the dominating pair of distributions for each step is given by $(P, Q)$, where $P = \mathcal{N}(1, \sigma^2)$ and $Q = \mathcal{N}(0, \sigma^2)$.

Applying standard composition results~\cite{zhu2022optimal, gopi2021numerical} to these dominating pairs, it follows that the $N$-fold adaptive composition of Algorithm~\ref{algo:noisy_pm} is itself dominated by the pair $(P, Q)$, where $P = \mathcal{N}(1, N \sigma^2)$ and $Q = \mathcal{N}(0, N \sigma^2)$. Thus, the resulting $(\epsilon,\delta)$-DP guarantee corresponds to that of the Gaussian mechanism with sensitivity 1 and noise parameter $\sqrt{N} \sigma$. The analytical expression from~\cite[Thm.~8]{balle2018improving} then yields the following bound.

\begin{lemma} \label{lem:gauss_dp}
Algorithm~\ref{algo:noisy_pm} is $\big(\epsilon,\delta(\epsilon) \big)$-DP for 
$$
    \delta(\epsilon) = \Phi\left( - \frac{\epsilon\sigma}{\sqrt{T}} + \frac{\sqrt{N}}{2\sigma} \right)
- e^\epsilon \Phi\left( - \frac{\epsilon\sigma}{\sqrt{N}} - \frac{\sqrt{N}}{2\sigma} \right),
$$
where $\Phi$ denotes the CDF of the standard univariate Gaussian distribution.
\end{lemma}

To compute $\epsilon$ as a function $\sigma$ and $\delta$ or $\sigma$ as a function of $\epsilon$ and $\delta$, the expression of Lemma~\ref{lem:gauss_dp} can be inverted using, e.g., bisection method.

\section{Experimental Comparisons}

Two of out of the three discussed algorithms perform well in practive: the graph perturtbation based mechanism and the noisy power method. 
Table~\ref{tab:error_bounds} summarized the derived error bounds for these two methods. We notice that the error bounds are similar except for the additional constant term in the bound for the graph perturbation mechanism. 

\subsection{Synthetic SBM Graphs}

These differences are also reflected in the experimental results shown in Fig.~\ref{fig:fig1},~\ref{fig:fig2}, and~\ref{fig:fig3} where we empirically measure the overlap rate for synthetic SBM graphs. Fixing the probabilities $p$ and $q$, we see that the noisy power method becomes the better one as the $n$ (number of nodes) increases.

In Fig.~\ref{fig:fig1},~\ref{fig:fig2}, and~\ref{fig:fig3}, each point describes the mean of 1000 experiments. The required $\sigma$-values for the noisy power method were computed using the expression given in Lemma~\ref{lem:gauss_dp}, based on $\epsilon$, $\delta$ and the number of iteration $N$ which was fixed to 8 for all experiments.





\begin{figure}[h!]
	\centering
    {\includegraphics[width=0.97\columnwidth]{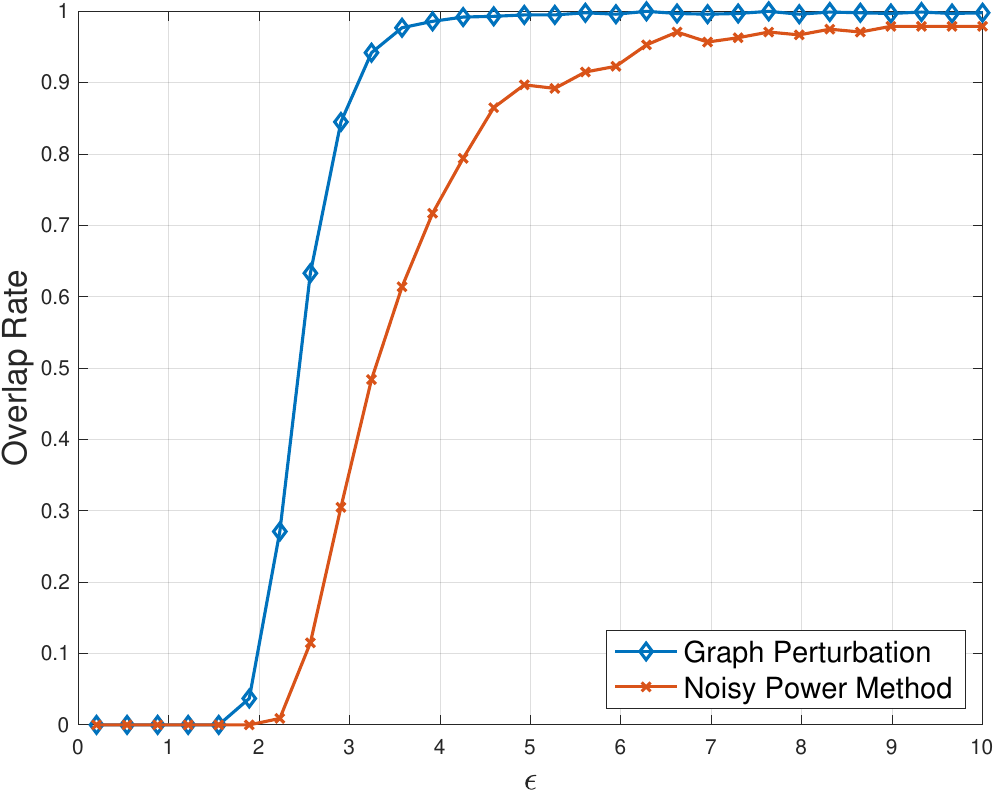}}
    \caption{\small{Overlap vs $\epsilon$, when $n=200$, $p =0.2$, $q = 0.02$, $\delta=n^{-2}$.  }}
    \label{fig:fig1}
\end{figure}

\begin{figure}[h!]
	\centering
    {\includegraphics[width=0.97\columnwidth]{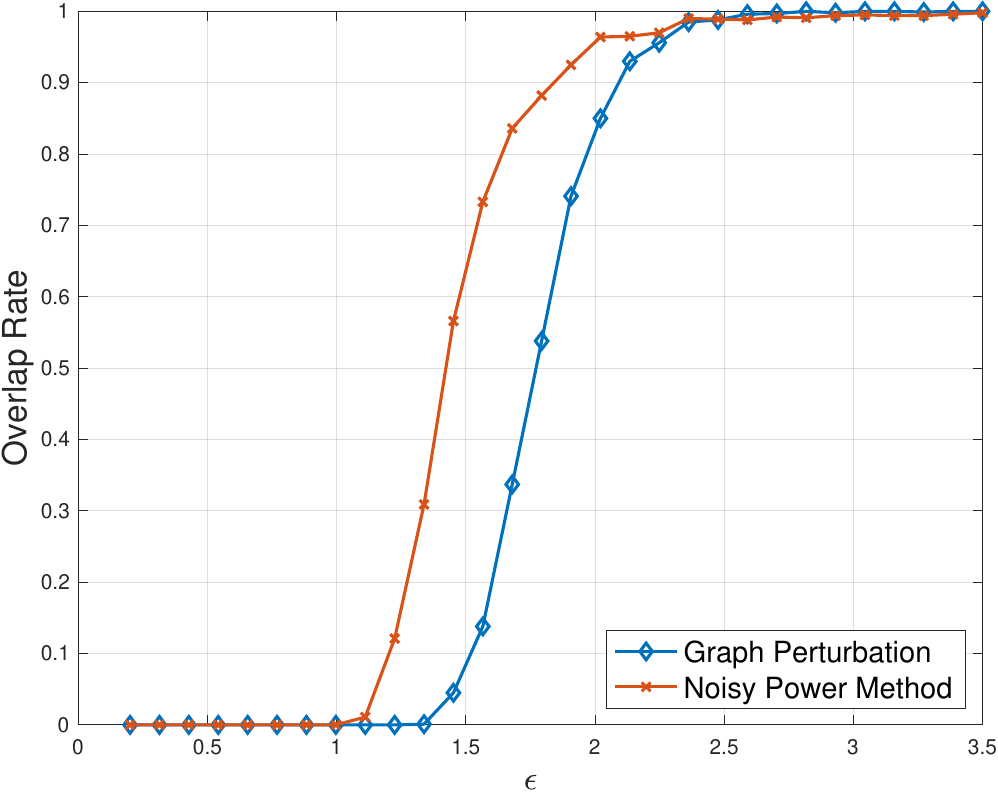}}
    \caption{\small{Overlap vs $\epsilon$, when $n=400$, $p =0.2$, $q = 0.02$, $\delta=n^{-2}$.  }}
    \label{fig:fig2}
\end{figure}

\begin{figure}[h!]
	\centering
    {\includegraphics[width=0.97\columnwidth]{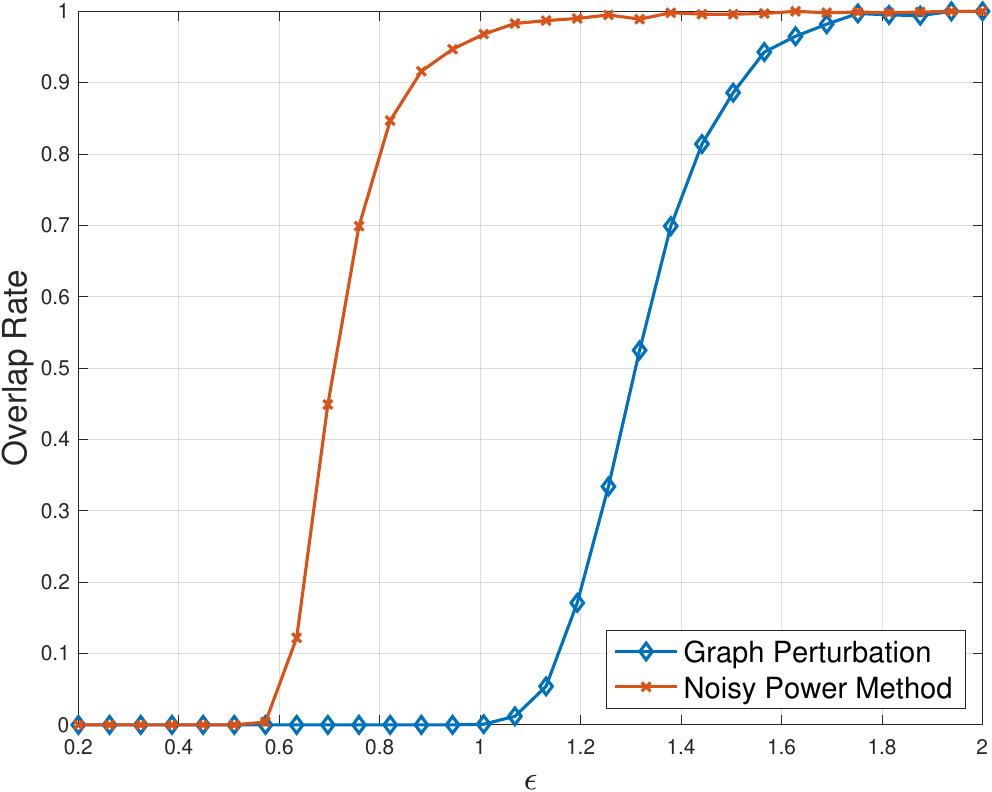}}
    \caption{\small{Overlap vs $\epsilon$, when $n=800$, $p =0.2$, $q = 0.02$, $\delta=n^{-2}$.  }}
    \label{fig:fig3}
\end{figure}

\subsection{Political Blogs Dataset}

We consider a symmetrized version of the Political Blogs dataset~\cite{adamic2005political}, originally a directed network of hyperlinks between U.S. political blogs collected in 2005. In this version, the direction of links is ignored, resulting in an undirected graph with 1,490 nodes and 16,718 edges. Each node represents a blog and is annotated with a label in $\{1,-1\}$ based on its content. Edge weights reflect the number of mutual hyperlinks between blog pairs, capturing the strength of their connection.

This dataset turns out to be more challenging for the noisy power method, as the eigenpair used for the deflation 
is not as good approximation of the leading eigenpair of the adjacency matrix $\mathbf{A}$ as in case of SBMs. We experimentally observe that only some of the randomly drawn initial vectors $\mathbf{y}_0$ for Algorithm~\ref{algo:noisy_pm} converge towards the second eigenvector of $\mathbf{A}$. To this end, we consider a variant, where we first privately search for a suitable initial vector by adding symmetric normally distributed $\sigma^2$ variance noise to $\mathbf{A}$ and setting $\mathbf{y}_0$ of Algorithm~\ref{algo:noisy_pm} to be the second eigenvector of the resulting noisy matrix. As a result, the total privacy guarantee can be seen as a $(N+1)$-wise decomposition of Gaussian mechanisms, each with noise scale $\sigma$, and we again get the $(\epsilon,\delta)$-DP guarantees using Lemma~\ref{lem:gauss_dp}.

Fig.~\ref{fig:fig_pol} shows the performance of the graph perturbation method and the noisy power method as a function of $\epsilon$, when $\delta=1/n^2$ for the noisy power method. In addition to the results for the fully private method where we privately initialize $\mathbf{y}_0$ for Algorithm~\ref{algo:noisy_pm}, we show the convergence with a randomly chosen initial value that converges to the second eigenvector of $\mathbf{A}$. 

In Fig.~\ref{fig:fig_pol}, each point describes the mean of 100 experiments. The number of iterations $N$ was fixed to 3 for the noisy power method.

\begin{figure}[h!]
	\centering
    {\includegraphics[width=0.97\columnwidth]{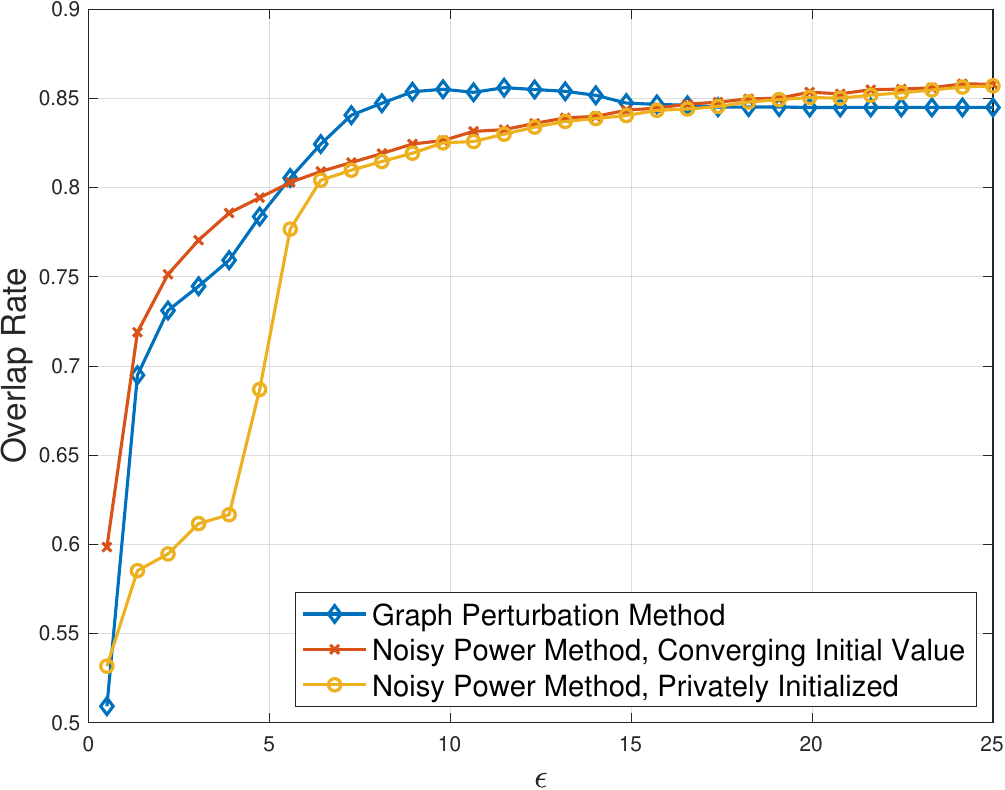}}
    \caption{\small{Overlap vs $\epsilon$ for the Political Blogs dataset.  }}
    \label{fig:fig_pol}
\end{figure}

\section{Conclusion}
\label{sec:conclusion}

We developed privacy-preserving spectral clustering methods for community detection over the binary symmetric SBMs under edge DP. Our approaches include $(1)$ a Graph Perturbation-Based Mechanism, which perturbs the adjacency matrix using either randomized response, followed by spectral clustering, $(2)$ a Subsampling Stability-Based Mechanism, which leverages subsampling and aggregation for accurate recovery, and $(3)$  an edge DP power method that adds carefully calibrated Gaussian noise to each matrix–vector multiplication, guaranteeing edge DP for every intermediate eigenvector estimate while still converging to the true leading eigenvectors. We also analyzed the tradeoff between privacy and accuracy, providing theoretical guarantees.  Future work will generalize these ideas to $(i)$ SBMs with more than two communities and $(ii)$ graphs exhibiting degree heterogeneity.

\bibliographystyle{IEEEtran}
\bibliography{myreferences}

\end{document}